\pdfoutput=1
\documentclass[a4paper,USenglish,cleveref]{lipics-v2021}
\hideLIPIcs
\nolinenumbers 
\usepackage{graphicx}
\usepackage{anyfontsize}
\usepackage{amsmath,amsthm,stmaryrd}
\SetSymbolFont{stmry}{bold}{U}{stmry}{m}{n}

\usepackage{amssymb}

\usepackage{multicol}
\usepackage[noend]{algpseudocode}
\usepackage{algorithm}
\usepackage{hyperref}
\hypersetup{pdfstartview=}
\usepackage{lipsum}
\usepackage{mathtools}

\usepackage{tikz-cd}
\usepackage{subcaption}

\DeclarePairedDelimiter\ceil{\lceil}{\rceil}
\DeclarePairedDelimiter\floor{\lfloor}{\rfloor}

\crefname{lemma}{\bf Lemma}{Lemmas}

\newcommand{\Bd}{B_{\sigma}}
\newcommand{\Sc}{S}
\newcommand{\pisub}{\pi}
\newcommand{\eps}{\epsilon}
\newcommand{\MB}{\mathsf{MB}}
\newcommand{\NB}{\mathsf{NB}}
\newcommand{\RegionH}{\Lambda_{\geq M}}
\newcommand{\OmegaBad}{\Omega^{\mathrm{bad}}}
\newcommand{\Vbase}{V_{\text{base}}}
\newcommand{\configpart}[1]{\sigma^{(#1)}}
\newcommand{\seqpart}[1]{\layerseq^{(#1)}}
\newcommand{\npart}[1]{n^{(#1)}}
\newcommand{\rtoppart}[1]{\rtop^{(#1)}}
\newcommand{\Bhatpart}[1]{\Bhat^{(#1)}}

\newcommand{\Sh}{S_h}
\newcommand{\nondec}{nondecelerating}
\newcommand{\colsc}{column scent}
\newcommand{\OmegaExt}{\overline{\Omega}}

\author{Shunhao Oh\footnote{Corresponding author.}}{School of Computer Science, Georgia Institute of Technology, Atlanta, GA, USA}{ohoh@gatech.edu}{https://orcid.org/0009-0002-1328-0040}{}
\author{Joseph L. Briones}{School of Computing and Augmented Intelligence, Arizona State University, Tempe, AZ, USA}{jbrione3@asu.edu}{https://orcid.org/0000-0002-5847-4263}{}
\author{Jacob Calvert}{Institute for Data Engineering and Science, Georgia Institute of Technology, USA}{calvert@gatech.edu}{https://orcid.org/0000-0001-9173-0946}{}
\author{Noah Egan}{School of Computer Science, Georgia Institute of Technology, Atlanta, GA, USA}{negan7@gatech.edu}{}{}
\author{Dana Randall}{School of Computer Science, Georgia Institute of Technology, Atlanta, GA, USA}{randall@cc.gatech.edu}{https://orcid.org/0000-0002-1152-2627}{}
\author{Andr\'ea W. Richa}{School of Computing and Augmented Intelligence, Arizona State University, Tempe, AZ, USA}{aricha@asu.edu}{https://orcid.org/0000-0003-3592-3756}{}

\begin{document}
\title{Single Bridge Formation in Self-Organizing Particle Systems}

\authorrunning{S. Oh, J.\,L. Briones, J. Calvert, N. Egan, D. Randall and A.\,W. Richa}

\Copyright{Shunhao Oh, Joseph L. Briones, Jacob Calvert, Noah Egan, Dana Randall and Andr\'ea W. Richa}

\keywords{Self-organizing particle systems, programmable matter, bridging, jump chain}

\ccsdesc[500]{Theory of computation~Distributed algorithms}


\maketitle

\begin{abstract}
Local interactions of uncoordinated individuals produce the collective behaviors of many biological systems, inspiring much of the current research in programmable matter.
A striking example is the spontaneous assembly of fire ants into ``bridges'' comprising their own bodies to traverse obstacles and reach sources of food. Experiments and simulations suggest that, remarkably, these ants always form one bridge---instead of multiple, competing bridges---despite a lack of central coordination. We argue that the reliable formation of a single bridge does not require sophistication on behalf of the individuals by provably reproducing this behavior in a self-organizing particle system. We show that the formation of a single bridge by the particles is a statistical inevitability of their preferences to move in a particular direction, such as toward a food source, and their preference for more neighbors. 
Two parameters, $\eta$ and $\beta$, reflect the strengths of these preferences and determine the Gibbs stationary measure of the corresponding particle system's Markov chain dynamics. We show that a single bridge almost certainly forms when $\eta$ and $\beta$ are sufficiently large. Our proof introduces an auxiliary Markov chain, called an ``occupancy chain,'' that captures only the significant, global changes to the system. Through the occupancy chain, we abstract away information about the motion of individual particles, but we gain a more direct means of analyzing their collective behavior. Such abstractions provide a promising new direction for understanding many other systems of programmable matter.
\end{abstract}

 \section{Introduction}
Ants are known to collaborate to form elaborate structures. Army ants of genus {\it Eiticon} self-assemble and disperse
while foraging to form shortcuts in order to more efficiently reach food \cite{Reid2015}. Weaver ants of species {\it Oecophylla longinoda} bind leaves together when building their nests and can form long chains \cite{Lioni2001}. In this work, we take inspiration from the behavior of fire ants of species {\it Solenopsis invicta}, which can self-assemble into floating rafts and structural bridges \cite{Mlot2011}. These bridges are structures composed of the ants as they entangle themselves and allow other ants a medium over which they can reach a target, such as food. It is not well understood why individual ants sacrifice themselves by becoming entangled with others in unfavorable circumstances such as suspension over water, but such occurrences are regularly accomplished without any centralized coordination. 

Recent laboratory experiments at the University of Georgia \cite{zeng2023fire} were the first to replicate the {more general}  
fire ant bridging found in nature by placing 
food in the center of a bowl filled with water to see how hungry ants will coordinate. Initially, ants surround the bowl looking for food, but soon they begin to explore the water, begin building multiple structures extending from the edge of the bowl. Eventually some structures start to reach the food, but over time only a {\em single bridge} persists, with the other ants traversing the bridge
to bring food back to the nest. 
Biologists have questioned what coordination allows the ants to always form a single bridge, streamlining their initial structures to minimize the number of ants sacrificed for the bridge while enabling the remaining ants to productively harvest food for the nest~\cite{Reid2015}. Physicists, biologists and computer scientists \cite{Cannon2016} have developed an agent-based model that mimics this behavior with locally interacting particles that primarily depends on two parameters:
{one  capturing their individual awareness of the food source, based on proximity, and another that captures the rigidity of the bridge structure by measuring an interparticle attraction among all nearest neighbor pairs, which has been shown to suffice to encourage collectives to aggregate \cite{Cannon2016}.} 
{This physical model also allows ants to climb over one another and incorporates additional physical parameters such as the meniscus effects at the edge of the bowl and effects of small clusters disconnecting.} Their simulations reliably replicate the tendency for exactly one bridge to form, but lack an explanation because the biological and simulated protocols are too difficult to analyze {in order to connect the local behavior of individuals with the complex coordinated structure formed by the ensemble.}

{{We show that the emergence of at most one bridge, as observed in ant experiments,
follows from a statistical inevitability, and not from the sophistication of the individuals.}} 
Here, rather than tracking the  movements of individual particles (or ants) over time according to a local Markov chain, 
we define a related auxiliary 
\emph{occupancy chain} that captures only the changes to which sites are occupied.
The states of the occupancy chain record the sites where one or more particles currently exist and the transitions 
capture the stochastic changes to the occupied sites over time, to determine long-term behavior. 
This approach of using occupancy chains allow us to apply techniques from statistical physics to infer properties of contours, and the simplified analysis is likely to be useful in other contexts, such as various fixed magnetization spin systems~\cite{dobrushin1992}, biased growth processes~\cite{grs}, and collectives arising in swarm robotics responding to directed external stimuli~\cite{li2021}.

\begin{figure}
 \centering
 \includegraphics[width=\textwidth]{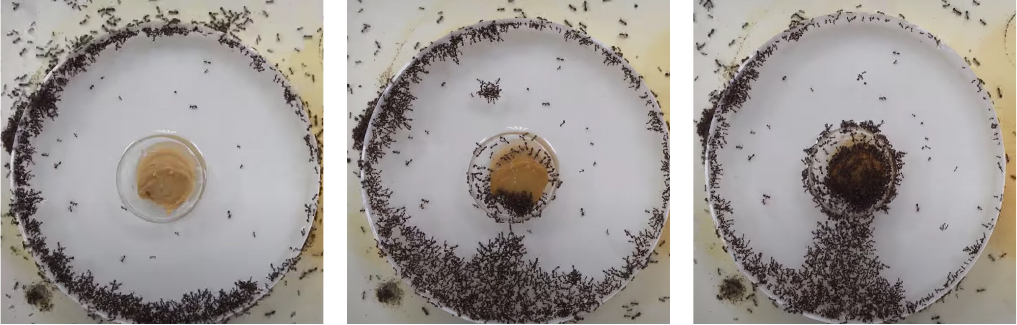}
 \caption{Photos of fire ants, $\textit{Solenopsis}$, self-assembling over time to form a bridge to reach food placed in the center of a bowl filled with water. Experiments performed by and photo credit to Takao Sasaki and Horace Zeng at the University of Georgia.}
 \label{fig:experiments}
\end{figure}

\subsection{Related work}
Ant bridging can be viewed as an example of programmable matter, in which researchers seek to create materials that can be programmed to change their physical properties in response to user input and environmental stimuli. Diverse examples of programmable matter---like DNA tiles \cite{wei2012}, synthetic cells~\cite{karzbrun2014}, and reconfigurable modular robots~\cite{daudelin2018}---are linked through their common use of local interactions to produce global functionality beyond the capabilities of individuals. Examples of this phenomenon abound in nature, including honey bees that select nest sites using decentralized recruitment~\cite{Camazine1999}, cockroach larvae that aggregate using short-range pheromones~\cite{Jeanson2005}, and fire ants that assemble their bodies into rafts to survive floods~\cite{Mlot2011}.  
For this reason, programmable matter is often designed by analogy with biological collectives that exhibit an emergent behavior of interest \cite{Sahin2005,ouellette2021}. 
However, this form of biomimicry is generally unamenable to rigorously proving that an algorithm executed by the constituents of matter reliably leads to a desired behavior. This analysis is critical to understanding the extent to which algorithms for programmable matter are robust when they are actually deployed \cite{li2021}.

An abstraction of programmable matter known as a {\it self-organizing particle system} (SOPS) addresses this gap. 
In the standard amoebot model for SOPS~\cite{Derakhshandeh2014, DaymudeRS21}, particles exist on the sites (or nodes) of a lattice, with at most one particle per site,  
and move between sites along lattice edges. Each particle is anonymous (unlabeled), interacts only with particles on adjacent lattice sites, has limited (normally assumed to be constant-size) memory, and does not have access to any global information such as a common direction on the lattice, a coordinate system or the total number of particles. We note that in this paper, we are studying a particle system with possibly multiple particles (ants) per site, but we reduce to an occupancy chain which recovers the ``SOPS perspective'' 
by just indicating whether each site is occupied or not. Recent work on SOPS has led to rigorous stochastic distributed algorithms for various tasks by exploring phase changes that lead to desirable collective behavior with high probability. Examples include compression and expansion~\cite{Cannon2016}, where particles cluster closely or expand; 
separation~\cite{Cannon2019}, where colored particles would like to separate into clusters of same color particles; aggregation~\cite{li2021}, where the system would like to compress but no longer needs to remain connected; shortcut bridging~\cite{Arroyo2018}; locomotion~\cite{Savoie2018,Yadav21}; and transport~\cite{li2021}. 
However, while these algorithms based on local Markov chains provide insight into emergent collective behavior, the proofs guaranteeing the long-term behaviors become prohibitively challenging as the collectives take on increasingly complicated tasks with more sophisticated interactions. For example, while the work in~\cite{Arroyo2018} also looked at the bridging behavior of fire ants~\cite{Reid2015}, it did not focus on the bridge formation process. Instead, the paper assumed that a single bridge already existed, anchored at two fixed points (corresponding to the nest and the food), and considered the (simpler) problem 
of optimizing the placement of the bridge with respect to parameters of relevance.

\subsection{Overview of our results}
\label{sec:results} 
To gain insight into structures that emerge from ant bridging, we consider a discrete model on a finite $w \times h$ domain $\Lambda=\Lambda(w,h)$ of the triangular lattice, with periodic boundary conditions in the horizontal direction, representing the circular bowl containing food. 
It will be convenient to think of $\Lambda$ as a cylindrical domain comprised of $h$ layers (or cycles), corresponding to the $h$ rows $\Lambda_1$ through $\Lambda_h$, each of width $w$  (see Figure~\ref{fig:lambda}). 
We model the motion of a finite set of $n$ particles entering $\Lambda$ at $\Lambda_1$,  
and growing a connected component by occupying sites on the interior of $\Lambda$ in order to eventually reach the food, which we assume can be found at every site on $\Lambda_h$.
Particles may move onto empty sites or walk upon sites currently occupied by other particles---we intentionally abstract away the specific dynamics by directly analyzing the occupancy chain, which only tracks changes to the set $\sigma$ of occupied sites in $\Lambda$.
In analogy with fire ant bridging, a step of the occupancy chain can be interpreted as a lone ant on the boundary of $\sigma$ climbing over a neighboring ant, thus creating an unoccupied site, or as an ant concluding its walk over other ants to occupy a new site adjacent to the previous boundary.
For simplicity, we assume that the aspect ratio $\alpha = w/h$ is a constant and that $n$ scales linearly with the area of the region, writing $n = \floor{\rho h^2}$ for some constant $\rho > 0$, as these are the conditions under which bridges may form.

\begin{figure}
 \centering
 \includegraphics[width=\textwidth]{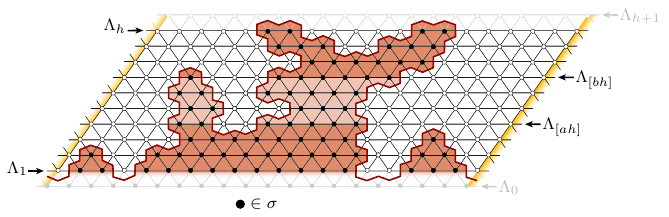}
 \caption{A domain $\Lambda$ of the triangular lattice of width of $w=20$ and height of $h=10$ with periodic boundary conditions identifying the ends of rows (indicated by yellow shading). The configuration $\sigma$ (black dots) has multiple $(a,b)$-bridges for $(a,b) = (0.4,0.7)$, but not for $(a,b) = (0.4,0.8)$.}
 \label{fig:lambda}
\end{figure}

The occupancy chain captures the relative preferences of particles to move toward $\Lambda_h$ (e.g., due to the scent of the food source placed along $\Lambda_h$), parameterized by $\eta >~0,$ and their preference for locations with more neighbors, increasing the rigidity of the connected structure. %
The balance of these preferences is expressed through a function on configurations called the {\em Hamiltonian}, which is related to the stationary probability of configurations and expresses the  overall ``fitness'' of a configuration, with each particle contributing its piece.

For a configuration $\sigma$ and a site $v \in \sigma$, let $\Sc(v)$ be the {\em intensity of the scent} at $v$, which 
is any nonincreasing function of graph distance $d(v)$ from $v$ to $\Lambda_h$.  
Let $\Bd(v)$ be the {\em number unoccupied sites that neighbor} $v$. 
Then the Hamiltonian~$H$ of a configuration $\sigma$ is
\begin{equation}\label{eq:hamiltonian}
    H(\sigma) = -\sum_{v \!\ \in \!\ \sigma} \left(\eta \Sc(v) - \Bd(v) \right),
\end{equation}
where $\eta > 0$ is a parameter describing how much a particle favors a stronger scent relative to how much it desires more neighbors. 
A second parameter $\beta$, the interparticle attraction, captures the strength of the Hamiltonian. 
These induce a probability  distribution $\pi$ over the set of valid configurations $\Omega$ (defined in Section~\ref{sec:model}) known as the Gibbs distribution, where 
\begin{equation}\label{eqn:stationary}
 \pi(\sigma) = \exp \left( -\beta H(\sigma) \right) / Z, \mbox{\ \ for $\sigma \in \Omega$}
\end{equation}
in terms of the normalization constant $Z = \sum_{\tau \in \Omega} e^{-\beta H(\tau)}$. 
 Note we have taken some liberties in defining the occupancy chain by assuming that each particle has been walking on the current bridge structure sufficiently long that it is near equilibrium. 

We may now talk about the shape of the contour between the set of occupied and unoccupied sites on $\Lambda$ in order to determine the number of bridges at equilibrium.  We introduce a weak definition of ``multiple bridges'' to show that any contours that unnecessarily ``backtrack'' a nontrivial amount from {\it any} layer will be extremely unlikely when the interparticle attraction $\beta$ is large enough. 
We say that a configuration $\sigma$ has {\em multiple $(a,b)$-bridges} for real numbers~$0<a<b<1$ if the restriction of $\sigma$ to layers $\floor{ah}$ and greater contains at least two components (on the lattice $\Lambda$) that reach layer $\floor{bh}$. This is formally stated later in Definition~\ref{defn:multiplebridges}.
Intuitively, one can imagine such configurations as ones that extend out to the layer $\Lambda_{\floor{bh}}$, retract back to $\Lambda_{\floor{ah}}$, then extend out to $\Lambda_{\floor{bh}}$ a second time.

Our first main theorem states that, on any $w \times h$ region $\Lambda$ with constant aspect ration $\alpha$, regardless of the scent parameter $\eta$, if the interparticle attraction $\beta$ is large enough, then for any constant $\epsilon \in (0,1)$, the {\em probability of seeing multiple $(a,b)$-bridges where $b-a \geq \epsilon$ is always exponentially small in $\beta$ and $h$} (see Theorem~\ref{thm:singlelongbridge} in Section~\ref{sec:formingbridges}). Taking $b=1$ shows that we are unlikely to see a contour that reaches the food multiple times while backtracking a distance $\epsilon$ between the intervals in which they touch.

Our second main result concerns whether or not any bridge will be likely to form.  This question is much more delicate and
depends on the scent parameter $\eta$, some basic assumptions about the scent gradient (see Definition~\ref{defn:nondec}), and some bounds on the number of particles in the system. 
If the {\em scent parameter $\eta$ is too small}, then particles will compress around the boundary and are unlikely to reach the food, so {\em no bridge will form}. If the {\em scent parameter is sufficiently large}, then a {\em single bridge will emerge, but never more}. These conditions are properly outlined in Section~\ref{sec:formingbridges} and stated as Theorem~\ref{thm:zeroorone-full} and Corollary~\ref{cor:scent}.

The proofs are motivated by a technique from statistical physics known as a Peierls argument that constructs a mapping from low-probability configurations to high-probability ones, to infer global properties of the distribution $\pisub$. The key challenge in constructing suitable mappings for our model is the dependence of the scent component $S$ of the Hamiltonian on the distance of sites from $\Lambda_h$. As a consequence, standard physics approaches based on simple mappings that rearrange sites is insufficient. We overcome this challenge by introducing the {\em layer sequence} of a configuration, which counts the number of the configuration’s elements in each layer. Its virtue is that it is far easier to describe and analyze transformations of layer sequences than those of configurations, and the layer sequence retains the information necessary to calculate the scent component of the configuration’s Hamiltonian. Moreover, while it discards the geometric information that is needed to calculate the configuration’s boundary length, we are nevertheless able to obtain useful bounds on the boundary length.

Our choice to model the motion of the contour directly via an occupancy chain, instead of inferring it from dynamics prescribed to individual particles (ants), {is a significant and valuable departure} from the norm in distributed computing. Without this abstraction tracking only the salient changes to the profile of the ensemble, it would be far more difficult to rigorously analyze the bridge formation behavior.
Notably, this abstraction comes at the cost of obscuring 
the underlying particles' computational capabilities and the local distributed algorithm that govern the particles' dynamics, which we address and clarify in Section~\ref{sec:algorithm}.

The success of SOPS in designing simple, local algorithms that produce complex, global behaviors suggests that the apparently intelligent behaviors of natural collectives may in fact be the inevitable consequences of unremarkable circumstances. We see a broader opportunity here, not only to explain how the intelligent behaviors of natural collectives arise, but also to  demonstrate that such an explanation need not assume great intelligence or detailed knowledge on behalf of the individuals. We believe incorporating occupancy chains may be a way to provide insight into other complex collectives.

\section{The Bridging Occupancy Chain}\label{sec:chain}
In the preceding section, we introduced the domain $\Lambda$, the Hamiltonian $H$, and the probability distribution $\pisub$ that $H$ induces, to the extent necessary to informally state our results. In particular, we mentioned that $\pisub$ is supported on a subset of valid configurations $\Omega$,
but deferred a definition. We now introduce an extension of $\Lambda$ that we use to define $\Omega$, and then we elaborate the terms of the Hamiltonian, before defining $\pisub$ and the occupancy chain. 

\subsection{The model}\label{sec:model}
Some subsets of the vertices of $\Lambda$ are not suitable abstractions of the bridges in Figure~\ref{fig:experiments}, because they do not extend from $\Lambda_1$ or because they contain holes. To make this precise, we extend the domain $\Lambda$ with two further layers, $\Lambda_0$ and $\Lambda_{h+1}$, the sites of which we respectively treat as always occupied and always unoccupied (Figure~\ref{fig:lambda}). We denote by $\overline{\Lambda}$ this extended domain, which is isomorphic to $\Lambda$ with a height of $h+2$, as shown in Figure~\ref{fig:lambda}. We define the set of valid configurations to be 
\[
\Omega = \left\{ \sigma \subseteq V(\Lambda): \text{ $|\sigma| \leq n$ \,and\, $\sigma \cup V(\Lambda_0)$ is simply connected in $\overline\Lambda$} \, \right\},
\] 
where $|\sigma|$ denotes the number of elements of $\sigma$.

Recall that the Hamiltonian (Equation~\ref{eq:hamiltonian}) of a configuration $\sigma \subseteq V(\sigma)$ involves a sum over sites $v \in \sigma$, with contributions from the scent intensity $\Sc(v)$ at $v$ and the number $\Bd(v)$ of neighbors of $v$ that are not in $\sigma$. To be precise, we define 
 $\Bd(v)$ according to 
\[
\Bd(v) = \left| \left\{u \in V(\Lambda)\setminus\sigma: (u,v) \in E(\Lambda) \right\} \right|,
\] 
which serves to penalize configurations that are less compact. Note that, if $\sigma \in \Omega$, then 
$\Bd(v) < 6$ for every $v \in \sigma$.
The scent intensity of a site $v \in V(\Lambda)$ is a function that depends only on its distance to the food, which we will assume is stronger the closer a site is to the food. 
For the sake of the proofs however, we will phrase these as distances $d(v,\Lambda_{0})$ from the bottom of the region $\overline\Lambda$, noting trivially that $d(v,V(\Lambda_{0}))=h-d(v,V(\Lambda_{h}))$. 
For convenience of notation, we define a nondecreasing function $\Sh: [h] \mapsto \mathbb{R}_{\geq 0}$, and write the scent intensity at a vertex $v$ as
\[
\Sc(v) = \Sh\big(d(v,V(\Lambda_{0}))\big).
\]
While the function $\Sh$ being nondecreasing is sufficient to show in Section~\ref{sec:nomultiplebridges} that multiple bridges are exponentially unlikely,
when we discuss whether or not a bridge to the food will form in Section~\ref{sec:formingbridges}, we use a slightly narrower class of functions that we call \emph{\nondec} (Definition~\ref{defn:nondec}).
This covers a broad class of functions, including but not limited to:
\[
\Sh(y) = C_h y^k - D_h \quad \text{or} \quad \Sh(y) = C'_h (h-y)^{-k'} - D'_h,
\]
where $k, k' \geq 1$, where the values $C_h$, $C'_h$ normalize the sum of $\Sh$ over a ``column'' of $\Lambda$ to a constant, and $D_h$, $D'_h$ are chosen to make $\Sh(1) = 0$. 
The purpose of the normalization is to make the contributions of $\Sc(v)$ and $\Bd(v)$ to the Hamiltonian comparable, so $\eta$ controls whether a bridge forms. In the case of $k,k'=1$, the first option gives us a linear scent gradient, while the second option represents a scent intensity proportional to the reciprocal of the distance from the food. Section~\ref{sec:formingbridges} discusses the cases when this sum is not normalized.

We define the occupancy chain in the next section and show that it constitutes an aperiodic and irreducible Markov chain on $\Omega$, hence it has a unique stationary distribution on $\Omega$, as given in Equation~\ref{eqn:stationary}.
Observe that by writing $\lambda = e^\beta > 0$ and $\gamma = e^{\beta \eta} > 0$, this formulation can be equivalently written as
\begin{equation}\label{eqn:stationaryweights}
 \pi(\sigma)
 \propto \prod_{v \in \sigma} \left(\lambda^{-\Bd(v)} \gamma^{\Sc(v)}\right),
\end{equation}
where the probability of each configuration can be written as a product of weights of its individual particles, normalized.

Although the Hamiltonian is partly inspired by the Ising model where the preference for having more neighbors can be thought of as the strength of a ferromagnetic interaction between particles, the scent component of the Hamiltonian introduces a directional bias that deviates from classical models, making the analysis considerably more challenging (see, e.g., \cite{grs}). Viewing the dynamics through the occupancy chain gives us a succinct tool to derive emergent behaviors of the collective not as easily provable for related agent-based models.

\subsection{The algorithm}\label{sec:algorithm}
The bridging occupancy chain algorithm can be briefly described as follows: Initially, the sites of $\Lambda$ are unoccupied. Given a configuration $\sigma$, we select a uniformly random site $v$ of $\Lambda$. If $v$ is occupied, then we attempt to remove it from $\sigma$. Otherwise, if $v$ is unoccupied and $|\sigma|$ is less than $n$, then we attempt to add $v$ to $\sigma$.
The move is rejected if it produces
a configuration that does not belong to $\Omega$, i.e., if it is
not simply connected in $\overline{\Lambda}$.

To enforce that simple connectivity is maintained at all times, starting from a simply connected configuration, we follow what we call \emph{local simple connectivity} checks, that can be shown to ensure that the configuration as a whole always remains simply connected in $\overline\Lambda$ (Lemma~\ref{lem:simplyconnected}).
Note that the checks below only depend on a site $v$ and its immediately adjacent sites in $\overline\Lambda$, which will be important when we discuss the underlying particle dynamics at the end of this section.

\begin{definition}[Local Simple Connectivity]
\label{dfn:localsimpleconnectivity}
For a site $v \in V(\Lambda)$, our definition of local simple connectivity depends on whether $v$ is occupied or not. Let $\overline{\mathcal{N}}(\sigma,v)$ denote the extended neighborhood of $v$ in $\sigma \cup V(\Lambda_0)$, where $v$ is any site of $\Lambda$.
\begin{enumerate}
\item If $v\in \sigma$, we say that $v$ is locally simply connected iff $|\overline{\mathcal{N}}(\sigma,v)| \leq 5$ and the induced subgraph $\overline{\Lambda}[\overline{\mathcal{N}}(\sigma,v)]$ is connected.
\item If $v\not\in \sigma$, we say that $v$ is locally simply connected iff 
$|\overline{\mathcal{N}}(\sigma,v)| \geq 1$ and 
the induced graph $\overline{\Lambda}[\overline{\mathcal{N}}(\sigma,v)]$ is connected.
\end{enumerate}
\end{definition}

\noindent The intuition behind this definition is that in Case (1) it is safe to remove a locally simply connected site $v$ from $\sigma$, since making $v$ unoccupied cannot disconnect its neighborhood in $\sigma$ (nor create a hole at $v$); in Case (2), making the site $v$ occupied cannot connect previously disconnected neighboring sites in $\sigma$ and therefore cannot create a hole. Note that, for every locally simply connected move, the reverse is also a valid move. 

Algorithm~\ref{alg:stimulusalgorithm} describes the moves of the occupancy chain.  Note that we are using a variant of the Metropolis-Hastings algorithm \cite{Metropolis1953}, but it is easy to check that for any two configurations whose symmetric difference is a single vertex, detailed balance  will be satisfied.

\begin{algorithm}[h!]
\begin{algorithmic}
\caption{Bridging Occupancy Chain}
\label{alg:stimulusalgorithm}
\State Let $\sigma$ denote the current configuration.
\State $v \gets $ uniformly random site in $\Lambda$
\State $p \gets $ uniformly random number in $[0,1]$
\If { $v$ is locally simply connected}
 \If {$v\in\sigma$ and $p \leq \min\{1, \exp\left(\beta(2\Bd(v) - \eta \Sc(v))\right)\}$}
 \State {Make $v$ unoccupied}.
\ElsIf {$v\not\in\sigma$, $|\sigma|<n$, and $p \leq \min\{1, \exp\left(-\beta(2\Bd(v) - \eta \Sc(v))\right)\}$}
 \State {Make $v$ occupied}.
 \EndIf
\EndIf
\end{algorithmic}
\end{algorithm}

The following two lemmas prove the correctness of the Bridging Occupancy Chain. 
We start by proving that the algorithm will keep the configurations simply connected and then show the Bridging Occupancy Chain is irreducible and aperiodic, and therefore converges to the indicated stationary distribution $\pi$.
%
 
\begin{lemma}[Maintaining Simple Connectivity]
\label{lem:simplyconnected}
Adding or removing a locally simply connected site from a 
configuration $\sigma\in\Omega$ maintains simple connectivity.
\end{lemma}

\begin{proof}
%
Let $\sigma$ be the current configuration, which we define as simply connected if $\sigma \cup V(\Lambda_0)$ is simply connected in $\overline{\Lambda}$. As $\Lambda_0$ is treated as always filled, $\sigma \cup V(\Lambda_0)$ is never empty.

The only way the removal of an occupied site $v\in \sigma$ can make the resulting configuration not simply connected is (i) if it introduces a hole, but this would imply that $v$ had six occupied neighbors in $\sigma\cup V(\Lambda_0),$ which is not allowed; or (ii) if it disconnects the configuration, but that would imply that the neighboring occupied sites to $v$ in $\sigma\cup V(\Lambda_0)$ were disconnected, which we also do not allow. On the other hand, the only way the addition of an occupied site $v$ to $\sigma$ can violate simple connectivity is if it introduces a cycle that was not present in $\sigma\cup V(\Lambda_0)$. This can only happen if the neighbors of $v$ in $\sigma\cup V(\Lambda_0)$ were either the empty set or if they induced a disconnected graph in $\sigma\cup V(\Lambda_0)$, which we again do not allow. The addition of an occupied site to $\sigma$ cannot disconnect the resulting configuration.
\end{proof}

\begin{lemma}[Irreducibility and Aperiodicity]
\label{lem:irreducibility}
The Bridging Occupancy Chain
is aperiodic and irreducible. Thus it converges to the stationary distribution $\pisub{}$ given by Equation~\ref{eqn:stationary}.
\end{lemma}

\begin{proof}
We first observe that for every locally simply connected move, the reverse is also valid and locally simply connected.
Next, we establish irreducibility of the occupancy chain by showing that from every configuration $\sigma \in \Omega$, we can always identify a site $v \in \sigma$ such that making $v$ unoccupied is a valid move. This would imply that we can always reach the empty configuration via valid moves, and therefore we can connect any pair of configurations $\sigma, \tau \in \Omega$ by first making all the sites in $\sigma$ unoccupied and then making all the sites in $\tau$ occupied by following the reversal of the path from $\tau$ to the empty configuration.  

Let $v$ be any site in $\sigma$ that has greatest distance from a site in $\Lambda_0$ following edges in $\overline{\Lambda}$ with both endpoints occupied.  We will argue that $v$ is locally simply connected.  Removing $v$ cannot form a hole because $v$ has greatest distance from $\Lambda_0$ and therefore cannot have six occupied neighbors. Likewise, making $v$ unoccupied cannot disconnect $\sigma$ or there would be occupied sites whose shortest path to $V(\Lambda_0)$ passes through $v$, contradicting $v$ having greatest shortest distance to $V(\Lambda_0)$.  Therefore $v$ is locally simply connected. 

Finally, we note that when $\sigma$ is nearly empty, most choices of $v$ will result in self-loops in the chain, so the chain is aperiodic.   Hence, since the chain is ergodic (aperiodic and irreducible),  it must have a unique stationary distribution, which is given by Equation~\ref{eqn:stationary} due to the Metropolis--Hastings transition probabilities satisfying detailed balance.
\end{proof}

Now that we have established the occupancy chain dynamics, it remains to discuss the underlying SOPS dynamics that would govern the behavior of the chain.
While the occupancy chain abstraction comes at the cost of obscuring the details of the underlying 
algorithm that each particle runs, there are indeed local distributed algorithms running independently on constant-memory particles, which we here assume to operate under a sequential scheduler (i.e., where at most one particle is active at any point in time, as in e.g.,~\cite{Cannon2016}),  that can closely approximate the occupancy chain dynamics. 
We first note that each iteration of the Bridging Occupancy Chain, and in particular the checks for local simple 
connectivity within it, can be executed by constant-sized memory
agents with only local communication,
since the algorithm only depends on the the states of the agents that occupy site $v$ and its immediate neighborhood in $\Lambda$. 

In the 
physics agent-based model for the fire-ant bridging biological experiments in~\cite{zeng2023fire} that were the original motivation of this work, 
each particle (ant) performs an independent random walk, biased by the scent gradient as it walks over other particles on the bridge, and biased by both the scent gradient and number of nearest neighbors as it leaves or joins the boundary.
The transition probabilities of the occupancy chain are the probabilities with which the bridge structure gains or loses one particle when this particle's walk on the current bridge structure is at stationarity. Moreover, while fine differences in the specific scent gradient affect the motion of the contour, this merely requires that individuals sense, and are acted upon by, the local 
gradient and does not require them to store and compute with correspondingly fine numbers. 
Thus, the occupancy chain model is rich enough to capture the long-term behavior of distributed collectives, 
and can also describe other underlying dynamics because it abstracts details that are irrelevant to our analysis.

\section{Equilibrium properties of the occupancy chain}\label{sec:theorems}
\newcommand{\Bhat}{\overline{B}}
\newcommand{\Shat}{\overline{S}}
\newcommand{\Rtop}{R^\top}
\newcommand{\LayerSet}{\mathcal{L}}
\newcommand{\layerseq}{\overline{N}}
\newcommand{\addable}{{addable}}
\newcommand{\removable}{{removable}}
\newcommand{\tight}{{tight}}
\newcommand{\rtop}{R_{top}}
\newcommand{\rbot}{R_{bot}}
\newcommand{\rectwidth}{w}

\newcommand{\thetaH}{\theta^{\text{high}}}
\newcommand{\thetaL}{\theta^{\text{low}}}

\newcommand{\addproof}[1]{\input{#1}}
We are now prepared to prove that at most one bridge will form by analyzing the stationary distribution of the occupancy chain. First, we argue that multiple bridges are unlikely if the affinity parameter is sufficiently large. Then, we identify a phase change in the formation of bridges as the scent parameter varies.

\subsection{Precluding the formation of multiple bridges}
\label{sec:nomultiplebridges}
In Section~\ref{sec:results} we briefly described a configuration with multiple $(a,b)$-bridges as one that extends out to the layer $\Lambda_{\floor{bh}}$, retracts back to $\Lambda_{\floor{ah}}$, then extends out to $\Lambda_{\floor{bh}}$ a second time. We now formally define this, as well as the \emph{depth} of the backtrack for such configurations.
\begin{definition}[Multiple $(a,b)$-bridges]
\label{defn:multiplebridges}
For real values~$0<a<b<1$, a configuration $\sigma \subseteq V(\Lambda)$ has {\em multiple $(a,b)$-bridges} if the subgraph of $\Lambda$ induced by $\bigcup_{k=\floor{ah}}^h \left(\sigma \cap V(\Lambda_{k})\right)$
contains at least two components that each contain some vertex from $\Lambda_{\floor{bh}}$. We call $\epsilon = b-a > 0$ the depth of the backtracking.
\end{definition}
To describe the rarity of multiple bridges forming, we consider the event that the configuration has multiple $(a,b)$-bridges for some pair $(a,b)$ with depth at least $\eps > 0$, defined as:
\begin{equation}\label{eqn:mbdef}
\MB_\eps = \cup_{a \in (0,1-\eps)} \left\{ \sigma \in \Omega: \text{$\sigma$ has multiple $(a,a+\eps)$-bridges} \right\}.
\end{equation}
We note that this event includes intervals of length greater than $\eps$ as a configuration with multiple $(a,a+\eps')$-bridges for $\eps' > \eps$ also has multiple $(a,a+\eps)$-bridges.
Our next result tells us that the probability of $\MB_\eps$ decays exponentially in $\beta$ and $h$, so long as $\beta$ is large enough in terms of the aspect ratio $\alpha$ and depth $\eps$. 

\begin{theorem}[No Multiple Bridges]\label{thm:singlelongbridge} 
For every $\eps > 0$ and every $w \times h$ region $\Lambda$ with aspect ratio $\alpha=w/h$,  there exists a positive number $\beta_0 = \beta_0 (\alpha, \eps)$ such that, if $\beta > \beta_0$, then the probability $\pi (\MB_{\eps})$ of  multiple bridges with depth $\eps$ is at most $e^{-\beta h \eps /2}$.
\end{theorem}

We emphasize that the quantity $\beta_0$ in Theorem~\ref{thm:singlelongbridge} does not depend on the constant $\rho$ (recall that $n = \floor{\rho h^2}$)
%
and scent parameter $\eta$, even though these parameters affect the stationary distribution. Instead, as the proof explains, $\MB_\eps$ is rare because its occurrence requires that the configuration has a relatively long boundary. While the number of configurations in $\MB_\eps$ grows exponentially in the perimeter of the domain $\Lambda$, this number is insufficient to compensate for the configurations' low weight when the boundary length penalty is high. Accordingly, for $\MB_\eps$ to be rare, it suffices for $\beta$ to be large in terms of $\alpha$, which controls the perimeter, and $\eps$, which controls the boundary length when $\MB_\eps$ occurs.

We now state three lemmas that we will use to prove Theorem~\ref{thm:singlelongbridge}. The first two of these results are due to the technical machinery of layer sequences, which we mentioned in Section~\ref{sec:results} and which we formally introduce in Section~\ref{sec:layerseq}. We note that, although we their proofs appear later, they do not rely on any earlier results. We state them in terms of the total boundary length and scent components of the Hamiltonian, which we define for a configuration $\sigma \in \Omega$ as
\[
B(\sigma) = \sum_{v \in \sigma} \Bd (v) \quad \text{and} \quad S(\sigma) = \sum_{v \in \sigma} \Sc (v).
\]

\begin{lemma}
\label{lem:lowperimeter}
For every $\sigma \in \Omega$, there exists $\tau \in \Omega$ such that 
$B(\tau) \leq \min\{ B(\sigma), 6w + 4h\}$
and $S(\tau) \geq S(\sigma)$.
\end{lemma}

Lemma~\ref{lem:lowperimeter} states that in the case where $\sigma$ has boundary length greater than $6w+4h$, we can always find another configuration $\tau$ that has boundary length at most $6w+4h$ with at least as great a scent component (we can set $\tau=\sigma$ otherwise).
This implies that $H(\sigma) \geq H(\tau)$ for all positive values of $\beta$ and $\eta$. Moreover, the boundary length of $\tau$ is bounded above in terms of $w+h$, instead of $w \times h$. The proof of Lemma~\ref{lem:lowperimeter} appears in Section~\ref{sec:layerseq}.

If $\sigma$ has multiple bridges, then we can further guarantee that there is a configuration $\tau$ with a boundary that is strictly shorter by an amount proportional to the depth of the bridges. This is the content of Lemma~\ref{lem:epsboundary}, which we prove in Appendix~\ref{app:proof}.

\begin{lemma} 
\label{lem:epsboundary}
For every $\eps > 0$ and $\sigma \in \MB_\eps$, there exists $\tau \in \Omega$ such that $B(\sigma) \geq B(\tau) + \eps h$.
\end{lemma}

Together, Lemmas~\ref{lem:lowperimeter}~and~\ref{lem:epsboundary} imply that a configuration $\sigma$ with multiple bridges can be mapped to a configuration $\tau$ that has far greater probability under $\pi$ because $H(\sigma) - H(\tau) \geq \beta \eps h$. This observation is one part of the proof of Theorem~\ref{thm:singlelongbridge}. The other part is a bound on the number of configurations with a given boundary length.

\begin{lemma}\label{lem:bdylength}
The number of configurations with a boundary length of $\ell$ is at most $2^{\ell+2w}$.
\end{lemma}

\begin{proof}
    Each configuration $\sigma \in \Omega$ of boundary length $\ell$ is fully described by a collection of self-avoiding paths of total length $\ell$ on the dual lattice. Each of these paths must start and end on either the top or bottom of $\Lambda$, giving $2w$ possible starting points. There are at most $2^{2w}$ ways to label each starting point as used or unused, and no more than $2^\ell$ ways to draw these paths given known starting points as there are at most two directions on the dual lattice to take each step of the path, giving an upper bound of $2^{\ell+2w}$.
\end{proof}

With these three lemmas, we are now prepared to prove our first main theorem:
%
\begin{proof}[Proof of Theorem~\ref{thm:singlelongbridge}]
We aim to upper bound the probability that the event $\MB_\eps$ occurs, using the fact that we can map the configurations it comprises to configurations with shorter 
boundaries. More precisely, by Lemma~\ref{lem:epsboundary}, for every $\sigma \in \MB_\eps$, there exists $\tau_\sigma \in \Omega$ such that 
\begin{equation}\label{eqn:tausigmabdy}
B(\sigma) - B(\tau_\sigma) \geq \eps h.
\end{equation}
In fact, we can assume that $\tau_\sigma$ further satisfies
\begin{equation}\label{eqn:tausigmascent}
    S(\sigma) \leq S(\tau_\sigma) \quad \text{and} \quad B(\tau_\sigma) \leq 6w + 4h
\end{equation}
because, if it did not, then we could use Lemma~\ref{lem:lowperimeter} to map $\tau_\sigma$ to a configuration that satisfies all three bounds in its place. We partition $\MB_\eps$ according to the difference in \eqref{eqn:tausigmabdy}:
\begin{equation}\label{eqn:partofmb}
\MB_\eps = \cup_{k \geq \eps h} \MB_{\eps,k}, \quad \text{where} \quad \MB_{\eps,k} = \left\{ \sigma \in \MB_\eps: B(\sigma) - B(\tau_\sigma) = k \right\}.
\end{equation}
The probability of each $\sigma \in \MB_{\eps,k}$ is exponentially small in $k$, due to the first bound in \eqref{eqn:tausigmascent}:
\begin{equation}\label{eqn:pisigmabd}
\pi (\sigma) \leq \frac{\pi (\sigma)}{\pi (\tau_\sigma)} = e^{-\beta (H(\sigma) - H(\tau_\sigma))} \leq e^{-\beta (B(\sigma) - B(\tau_\sigma))} = e^{-\beta k}.
\end{equation}
This bound is useful because $|\MB_{\eps,k}|$ is at most exponentially large in $k$. Indeed, due to the second bound in \eqref{eqn:tausigmascent}, every $\sigma \in \MB_{\eps,k}$ has boundary length of at most 
\[
B(\sigma) = B(\tau_\sigma) + k \leq 6w + 4h + k.
\]
There are at most $2^{2w+\ell}$ configurations with a boundary length of exactly $\ell$ (Lemma~\ref{lem:bdylength}), hence the number of configurations in $\MB_{\eps,k}$ is at most
\begin{equation}\label{eqn:mbcard}
|\MB_{\eps,k}| \leq |\{\sigma \in \Omega: B(\sigma) \leq 6w + 4h + k\}| \leq \sum_{\ell \leq 6w + 4h + k} 2^{2w+\ell} \leq 4^{7w + 4h + k}.
\end{equation} 
We combine the partition of $\MB_\eps$ with the estimates of $\pi (\sigma)$ and $|\MB_{\eps,k}|$ to conclude that 
\[
\pi (\MB_\eps) \stackrel{\eqref{eqn:partofmb}}{=} \sum_{k \geq \eps h}  \sum_{\sigma \in \MB_{\eps,k}} \pi (\sigma) \stackrel{\eqref{eqn:pisigmabd}}{\leq} \sum_{k \geq \eps h}  |\MB_{\eps,k}| e^{-\beta k} \stackrel{\eqref{eqn:mbcard}}{\leq} \sum_{k \geq \eps h} 4^{7w+4h+k} e^{-\beta k}.
\]
Since $w = \alpha h$, this bound is at most $e^{-\beta h \eps/2}$ when $\beta$ is large enough in terms of $\alpha$ and $\eps$. 
\end{proof}

\subsection{Conditions for bridge formation}
\label{sec:formingbridges}

\newcommand{\layerseqbot}{{\layerseq}^{\bot}}
\newcommand{\layerseqtop}{{\layerseq}^{\top}}
\newcommand{\nbot}{{n}^{\bot}}
\newcommand{\ntop}{{n}^{\top}}
\newcommand{\nfull}{{n}^{\text{full}}}
\newcommand{\Nbotsum}{{N}^{\bot}}
\newcommand{\Ntopsum}{{N}^{\top}}
\newcommand{\Nfullsum}{{N}^{\text{full}}}
\newcommand{\toplayer}{D}

\newcommand{\colbot}{{\mathcal{C}}^{\bot}}
\newcommand{\coltop}{{\mathcal{C}}^{\top}}

\newcommand{\extminus}{a_{(-)}}
\newcommand{\extplus}{a_{(+)}}

\newcommand{\Scont}{\widehat{S}_h}
\newcommand{\nin}{{|\sigma|}}

\newcommand{\numfullcols}{J}
\newcommand{\Jfull}{J^{\text{full}}}
\newcommand{\Jbot}{J^\bot}

\newcommand{\spre}{\text{pre}}
\newcommand{\stran}{\text{mid}}
\newcommand{\spost}{\text{post}}

\newcommand\numberthis{\addtocounter{equation}{1}\tag{\theequation}}

In this section, we show that whether a bridge to the food forms depends on the balance
between the scent and boundary length components of the Hamiltonian. Recall that $\eta$ controls each particle's desire to move toward the food relative to its desire for more neighbors.

The competition between scent and boundary length is richest when there are not too few or too many particles in the system. When there are too few particles, no bridges can form; when there are too many particles, ``bridges'' can form without incurring a high boundary cost, due to the periodic boundary conditions. Accordingly, we assume that the number of particles $n$ is $\floor{\rho h^2}$, for a constant
$\rho$ that satisfies $1/2 < \rho < \alpha - 2$, in terms of the aspect ratio $\alpha = w/h$ of the domain $\Lambda$. Implicitly, we require that $\alpha$ is greater than $2.5$. 

We write the {{scent intensity}} as a function $\Sh: [h] \to \mathbb{R}_{\geq 0}$, where $\Sh(k)$ represents the {intensity} of any particle on layer $k$.
Previously, we only assumed that $\Sh$ was a nondecreasing function (i.e., nonincreasing in the distance from the food), and showed that as long as $\beta$ is sufficiently high, at most one bridge forms. Now we assume additional light conditions on $\Sh$, which we believe will apply for most ``natural'' definitions of the scent intensity function. We use the term \emph{\nondec{} functions} to refer to this broad class of functions.
\begin{definition}
\label{defn:nondec}
For a function $\Sh: [h] \to \mathbb{R}_{\geq 0}$, let $\Delta \Sh(k) = \Sh(k+1) - \Sh(k)$ denote the discrete derivative of $\Sh$ at $k \in [h-1]$. We say that $\Sh$ is {\em \nondec{}} if $S_h(1) = 0$, $\Delta S_h (k) \geq 0$ for all $k \in [h-1]$, and $\Delta (\Delta \Sh)(k) \geq 0$ for all $k \in [h-2]$.
\end{definition}

Assuming that the scent function is nondecelerating, our second main result identifies a phase change in the probability that no bridge reaches the food. We define this event as
\[
\NB = \{ \sigma \in \Omega: \sigma \cap V(\Lambda_h) = \emptyset \}.
\]
The phases correspond to ranges of $\beta$ and $\eta$ that depend on the \emph{\colsc{}} $\varphi = \sum_{k=1}^h \Sh(k)$, which represents the sum of the scent intensities over a single ``column'' of $\Lambda$. We consider the case when $\varphi > 0$ is a constant with respect to $h$, the aspect ratio satisfies $\alpha>2.5$, and the 
%
constant $\rho$ belongs to $(0.5,\alpha-2)$. We define the phases in terms of the quantities
\[
\beta_1 = \frac{2\rho + 3 + 4\alpha}{2\rho - 1}\log 2, \quad \beta_2 = \left(1 + \frac{\alpha}{2}\right)\log 2, \quad \text{and} \quad \eta_1 = \frac{4}{\varphi}\left(1 + \frac{1}{\rho}\right),
\]
as well as $\eta_2 = \eta_2 (\beta)$, given by
\[
\eta_2 = \frac{1}{\varphi}\min\left\{ 2\left(1 - \frac{\log 2}{\beta}\right), \frac{4}{1+\rho}\left( 1 - \left(1 + \frac{\alpha}{2}\right)\frac{\log 2}{\beta} \right) \right\}.
\]
Note that both $\eta_1$ and $\eta_2$ are proportional to $1/\varphi$, and $\eta_2$ is strictly positive when $\beta > \beta_2$.

\begin{theorem}[Phase Change]\label{thm:zeroorone-full}
Suppose that
$1/2 < \rho < \alpha - 2$ and let $\Sh$ be \nondec{}.
\begin{enumerate}[(i)]
     \item {\sf{(At Least One Bridge)}} If $\beta > \beta_1$ and $\eta > \eta_1$, then there exist positive numbers $c_1 = c_1 (\beta,\eta)$ and $h_1 = h_1 (\beta,\eta)$ such that
$\pi (\NB) \leq e^{-c_1 h}, $ for all  $h \geq h_1$.  
\vskip.05in
\item  {\sf{(No Bridge)}} If $\beta > \beta_2$ and $\eta < \eta_2(\beta)$, then there exist positive numbers $c_2 = c_2 (\beta,\eta)$ and $h_2 = h_2 (\beta,\eta)$ such that
$
\pi (\NB) \geq 1 - e^{-c_2 h}, $ for all $ h \geq h_2$.
\end{enumerate}
\end{theorem}

Since $\eta_1$ and $\eta_2$ are proportional to $1/\varphi$, an immediate corollary of Theorem~\ref{thm:zeroorone-full} is that, if the scent is especially strong or especially weak, in the sense that the asymptotic growth of the column scent as $h \to \infty$ satisfies $\varphi = \omega_h (1)$ or $\varphi = o_h (1)$, then bridge formation is either likely or unlikely, independently of $\eta$.

\begin{corollary}
\label{cor:scent}
Suppose that $1/2 < \rho < \alpha - 2$ and $\Sh$ is \nondec{}.
\begin{enumerate}[(i)]
     \item {\sf{(Strong Scent)}} If $\beta > \beta_1$ and $\varphi = \omega_h(1)$, then there exist positive numbers $c_3 = c_3(\beta)$ and $h_3 = h_3(\beta)$ such that $\pi (\NB) \leq e^{-c_3 h}$, for all $h \geq h_3$.  
\vskip.05in
\item  {\sf{(Weak Scent)}} If $\beta > \beta_2$ and $\varphi = o_h(1)$, then there exist positive numbers $c_4 = c_4(\beta)$ and $h_4 = h_4(\beta)$ such that $\pi (\NB) \geq 1 - e^{-c_4 h}$, for all $h \geq h_4$.
\end{enumerate}
\end{corollary}

Note Corollary~\ref{cor:scent}(ii) holds without the assumption that $\Sh$ is nondecelerating. In fact, it holds when $\Sh$ is nonnegative because, when $\varphi = o_h (1)$, the contribution of boundary length to the Hamiltonian dominates the scent component, making fine properties of $\Sh$ irrelevant.

The proof of Theorem~\ref{thm:zeroorone-full} relies on three lemmas. Informally, the first says that a subset of configurations $\OmegaBad$ has exponentially small probability in $h$ if every $\sigma \in \OmegaBad$ can be mapped to a configuration with a Hamiltonian that is smaller by roughly $h + B(\sigma)$. The proof, which can be found in the full version of the paper \cite{fullpaper}, is similar to that of Theorem~\ref{thm:singlelongbridge}.

\begin{lemma}\label{lem:unlikelyconfigurations}
Fix an aspect ratio $\alpha$, as well as $\delta$ and $\epsilon$, all positive numbers. Suppose that $\OmegaBad$ is a subset of configurations such that, for every $\sigma \in \OmegaBad$, there exists $\tau_\sigma \in \Omega$ satisfying 
\[
H(\sigma) - H(\tau_\sigma) \geq \epsilon h + \delta B(\sigma) + O_h(1),
\]
where the function implicit in $O_h(1)$ does not depend on $\sigma$. Then, for every 
\[
\beta > \beta_0 (\delta, \eps) := \max\left\{\frac{\log 2}{\delta}, \frac{2\alpha \log 2}{\epsilon}\right\},
\]
there exist $c = c(\beta) > 0$ and $h_0 = h_0 (\beta) > 0$ such that $\pi (\OmegaBad) \leq e^{-ch}$ for all $h \geq h_0$.
\end{lemma}

The next lemma  states that if a configuration does not form a bridge, then it is possible to identify a configuration with a Hamiltonian that is strictly smaller as a function of $h$ and the boundary length of the original configuration.

\begin{lemma}\label{lem:thm2case1}
If $\beta > \beta_1$ and $\eta > \eta_1$, then there exist positive values $\delta_1=\delta_1(\rho,\alpha)$ and $\eps_1=\eps_1(\rho,\alpha)$ such that $\beta_1 \geq \beta_0 (\delta_1,\eps_1)$ and, for every $\sigma \in \NB$, there exists $\tau_\sigma \in \Omega$ satisfying
\[
H(\sigma) - H(\tau_\sigma) \geq \eps_1 h + \delta_1 B(\sigma) + O_h (1).
\]
\end{lemma}

The proof involves a multi-stage transformation of the original configuration's layer sequence (Appendix~\ref{app:proof}). The same is true of the next result, which instead compares configurations that reach $\Lambda_h$ to the empty configuration $\emptyset$, which has $H(\emptyset) = 0$.

\begin{lemma}\label{lem:thm2case2}
If $\beta > \beta_2$ and $\eta < \eta_2 (\beta)$, then there exist positive values $\delta_2=\delta_2(\eta,\varphi,\rho,\alpha)$ and $\eps_2=\eps_2(\eta,\varphi,\rho,\alpha)$ such that $\beta_2 \geq \beta_0 (\delta_2,\eps_2)$ and every $\sigma \notin \NB$ satisfies
\[
H(\sigma) - H(\emptyset) \geq \eps_2 h + \delta_2 B(\sigma) + O_h (1).
\]
\end{lemma}

Theorem~\ref{thm:zeroorone-full} follows directly from the preceding lemmas. 

\begin{proof}[Proof of Theorem~\ref{thm:zeroorone-full}]
Parts (i) and (ii) follow from two applications of Lemma~\ref{lem:unlikelyconfigurations}, one to $\OmegaBad = \NB$ and another to $\OmegaBad = \NB^c$, which are justified by Lemmas~\ref{lem:thm2case1}~and~\ref{lem:thm2case2}.
\end{proof}

\subsection{Layer sequences}\label{sec:layerseq}

The proofs of our main results rely on mappings of configurations to ones with more desirable properties (Lemmas~\ref{lem:lowperimeter},~\ref{lem:epsboundary},~\ref{lem:thm2case1},~and~\ref{lem:thm2case2}). For example, Lemma~\ref{lem:lowperimeter} states that every configuration can be mapped to a second configuration with at least as favorable boundary length and scent components, and which has a boundary length of $O(w+h)$. The purpose of this section is to introduce layer sequences, the key technical idea underlying these lemmas, and to demonstrate their use by proving Lemma~\ref{lem:lowperimeter}.

The layer sequence $\layerseq = (n_k)_{k = 1}^h$ of a (possibly non-valid) configuration $\sigma \subseteq V(\Lambda)$ counts the number $n_k = n_k (\sigma)$ of its elements in each layer of $\Lambda$: 
\[
n_k(\sigma) = |\sigma \cap V(\Lambda_k)|.
\]
The statements that we will make in this section will apply not only to configurations in $\Omega$, but a broader class $\OmegaExt$ of configurations the layer sequences of which have no fully occupied layer after one that is not fully occupied (i.e., no $k$ such that $n_k(\sigma)=w$ and $n_{k-1}(\sigma)<w$), and no fully unoccupied layer before one that is not fully unoccupied (i.e., no $k$ such that $n_k(\sigma)=0$ and $n_{k+1}(\sigma)>0$):
\begin{align*}
\OmegaExt = \big\{ \sigma \subseteq V(\Lambda) :~&\left(n_{k+1}(\sigma) =w \implies n_k(\sigma) = w\right) \text{ and }\\
&\left(n_k(\sigma)=0 \implies n_{k+1}(\sigma)=0\right) \text{ for all } k \in [h-1] \big\}.
\end{align*}
Configurations in $\Omega$ must have this property as they are simply connected, hence $\Omega \subseteq \OmegaExt$.
Referencing this property, for $0\leq \rbot\leq \rtop \leq h$, we can denote by $\LayerSet^{\rbot,\rtop}$ the set of layer sequences that
are fully occupied from layers $1$ to $\rbot$, partially occupied from layers $\rbot+1$ to $\rtop$, and fully unoccupied for the remaining layers.

Even though the layer sequence of a configuration discards information about the configuration's boundary, we can still use it to obtain bounds on the boundary length (our definition of boundary length extends to configurations in $\OmegaExt$). This is the content of the next two results. To state them, we denote by $\sigma_M$ the part of a configuration $\sigma \in \OmegaExt$ in layers $1$ through $M$:
\[
\sigma_M = \cup_{k=1}^M (\sigma \cap V(\Lambda_k)).
\]

\newcommand{\Mmin}{R_\text{min}}
\begin{lemma}[Boundary length increment] 
\label{lem:boundarylengthincrement}
Let $0 \leq \rbot \leq \rtop \leq h$ and let $\sigma \in \OmegaExt$ have layer sequence~$\layerseq \in \LayerSet^{\rbot,\rtop}$.
For integers $M$ where $\max\{\rbot,1\}\leq M \leq \min\{\rtop-1,h-2\}$,
the truncation $\sigma_{M+1}$ has boundary length satisfying $B(\sigma_{M+1}) \geq B(\sigma_M) + D_M$, where 
\[
D_M = 2 + 2\max\left\{n_{M+1}-n_M+1,0\right\} + 2\max\left\{n_{M+1}-n_M-1,0\right\}.
\]
If $M = h-1$ instead, then we have $B(\sigma_{M+1}) = B(\sigma) \geq B(\sigma_M) + D_M - 2n_h$. Furthermore, the bound is tight in the sense that, for any $\layerseq \in \LayerSet^{\rbot,\rtop}$, there exists $\tau \in \Omega$ (not $\OmegaExt$) with layer sequence $\layerseq$ which satisfies $B(\tau_{M+1}) \leq B(\tau_M) + D_M$ for every such $M$.
\end{lemma}

The proof of Lemma~\ref{lem:boundarylengthincrement}, which is self-contained, appears in Appendix~\ref{app:proof}. Lemma~\ref{lem:boundarylengthincrement} implies a lower bound on a configuration's boundary length in terms of its layer sequence.

\begin{lemma}[Boundary length lower bounds] 
\label{lem:boundarylengthlowerbound}
Suppose that $\layerseq \in \LayerSet^{\rbot,\rtop}$ where $0 \leq \rbot \leq \rtop \leq h$ and let $\Mmin = \max\{\rbot,1\}$. For each $M \in \{\Mmin,\Mmin+1,\ldots,\rtop\}$, we define
\[
\Bhat_M(\layerseq) = \Bhat_{\Mmin}(\layerseq) + \sum_{k=\Mmin}^{M-1} D_k, \quad\,\,
\text{where}
\quad\,\, \Bhat_{\Mmin}(\layerseq) = \begin{cases}2n_1+2 &\text{ if }\rbot=0, \\2\rectwidth &\text{ if }\rbot \geq 1. \end{cases}
\]
Then, for any configuration $\sigma \in \OmegaExt$ with layer sequence $\layerseq$, for any $M \in \{\Mmin,\Mmin+1,\ldots,\rtop\}$, the truncation $\sigma_M$ has boundary length satisfying
\[B(\sigma_M) \geq \begin{cases}\Bhat_M(\layerseq) &\text{if }M<h,\\ \Bhat_M(\layerseq) - 2n_h &\text{if }M=h.\end{cases}\]

Furthermore, this bound is tight in the sense that for any layer sequence $\layerseq \in \LayerSet^{\rbot,\rtop}$, there exists a configuration $\tau \in \Omega$ with layer sequence $\layerseq$ where $B(\tau_{\rbot+1}) = \Bhat_{\rbot+1}(\layerseq)$ and $B(\tau_{M+1}) \leq B(\tau_M) + D_M$ for each $M \in \{\Mmin,\Mmin+1,\ldots,\rtop\}$.
\end{lemma}

\begin{proof}[Proof of Lemma~\ref{lem:boundarylengthlowerbound}]
This follows from Lemma~\ref{lem:boundarylengthincrement} by induction on $M$, with $M = \Mmin$ as the base case.
\end{proof}

We use the two preceding facts about layer sequences to prove Lemma~\ref{lem:lowperimeter}. The proof features the quantity $\Shat(\layerseq)$, which is the scent intensity common to all configurations with layer sequence $\layerseq$.

\begin{proof}[Proof of Lemma~\ref{lem:lowperimeter}]
We aim to show that, for an arbitrary configuration $\sigma \in \Omega$, there exists $\tau \in \Omega$ such that $B(\tau) \leq B(\sigma)$, $S(\tau) \geq S(\sigma)$, and $B(\tau) \leq 6w + 4h$. The idea of the proof is to show that, if there are layers $k_1$ and $k_2$ such that the layer sequence $\layerseq$ of $\sigma$ satisfies
\[
k_2 \geq k_1+2, \quad n_{k_1+1} - n_{k_1} \geq 2, \quad \text{and} \quad n_{k_2+1} - n_{k_2} \geq 2,
\]
then it is possible to ``promote'' one of the occupied sites to a higher layer, resulting in a new layer sequence $\layerseq^+ = (n^+_k)_{k \in [h]} \in \LayerSet^{\rbot^+,\rtop^+}$, in such a way that $\Bhat_h(\layerseq^+) \leq \Bhat_h(\layerseq)$. Moving an occupied site to a higher layer clearly also ensures that $\Shat(\layerseq^+) \geq \Shat(\layerseq)$, because $S_h$ is nonincreasing with distance from the food. 
We define $\layerseq^+$ by removing an occupied site from layer $k_1+1$ and adding it to layer $k_2$. In other words, for each $k \in [h]$, we define: $n^+_k = n_k - 1$ if $k = k_1 + 1$; $n^+_k = n_k + 1$ if $k = k_2$; and $n^+_k = n_k$ otherwise. 
Note that, like $\rbot$, $\rbot^+$ is at most $k_1$.

To compute the difference between $\Bhat_h(\layerseq^+)$ and $\Bhat_h(\layerseq)$, we look at the terms affected by the change. For each $k$, denote by $D_k$ and $D_k^+$ the terms of the sums of $\Bhat_h(\layerseq)$ and $\Bhat_h(\layerseq^+)$. We observe that $D^+_{k_1} = D_{k_1} - 4$ and $D^+_{k_2} = D_{k_2} - 4$, and the differences $D^+_{k_1+1} - D_{k_1+1}$ and $D^+_{k_2-1} - D_{k_2-1}$ are at most $8$ when $k_1+1$ and $k_2-1$ are equal, and at most $4$ when they are not. This implies that $\Bhat_h(\layerseq^+) \leq \Bhat_h(\layerseq) - 8 + 8 = \Bhat_h(\layerseq)$.

We repeat this process of promoting occupied sites to find a layer sequence $\layerseq^* = (n^*_k)_{k \in [h]}$ such that $\Bhat_h(\layerseq^*) \leq \Bhat_h(\layerseq)$, $\Shat(\layerseq^*) \geq \Shat(\layerseq)$, and $n^*_{k+1} - n^*_k \leq 1$ for all but at most two values of $k$. Moreover, if there are two such values of $k$, they must be consecutive integers.
%
This implies the upper bound
\begin{align*}
\Bhat_h(\layerseq^*) \leq \Bhat_{\rbot^*+1}(\layerseq) + 4(\rtop^*-\rbot^*) + 4w \leq 6w + 4h.
\end{align*}
By Lemma~\ref{lem:boundarylengthlowerbound}, there exists $\tau \in \Omega$ with layer sequence $\layerseq^*$ and $B(\tau) \leq \Bhat_h(\layerseq^*)$, hence $B(\tau) \leq 6w + 4h$. 
\end{proof}

\section{Conclusion}
%
We define a simple SOPS model of bridging based on the particles' affinity for more neighbors, which results in more ``robust'' bridges,  and a bias toward the food. We show that the emergence of a single bridge in collective systems  is
a statistical inevitability, requiring no central coordination.
The novelty of our strategy is based on defining and analyzing a much simpler occupancy chain that abstracts the specific local dynamics of the particles and looks at the evolution of the contour indicating the system's occupied sites.
We intentionally kept the occupancy chain as simple as possible to show the generality of the single bridging emergent phenomena to highlight the connections to similar statistical physics models.\footnote{The techniques we use also apply to other geometries, including the square lattice and other planar graphs with or without periodic boundary conditions.} 
Moreover, the occupancy chain that evolves according to a Hamiltonian defined only by contour length and scent gradient is simple enough to provide rigorous proofs using tools from statistical physics.

We see several
directions for future work: First, we believe that we can extend our results to other models inspired by biological systems and beyond, such as by including ants retreat from the food after they are fed, by using a more sophisticated occupancy chain that models rafts by allowing a collection of ants to move together in a single move, or by including more specifics of man-made swarm robotics systems. Our simulations show that such variants do not significantly change our findings. 
Second, the occupancy chain abstracts away information about the motion of individual particles, and allows us to
gain a more direct means of analyzing their collective behavior: We expect that such abstractions will help understanding the collective behavior of many other programmable matter systems. Third, occupancy Markov chains have the potential for impact in statistical physics, by allowing one to relax the need for precisely estimating surface tension of contours, potentially enabling a better formal understanding of fixed magnetization spin systems \cite{dobrushin1992}, and collectives arising in swarm robotics responding to directed external stimuli \cite{li2021}.  Lastly, while simulations suggest that both the occupancy chain and the underlying agent-based simulations converge in polynomial time, we do not have a formal bound on the mixing time of either.  Future work can relax the connectivity restriction on valid configurations, which may make it easier to derive such bounds, but ant experiments suggest typically the majority do stay connected.  It would also be helpful to derive general bounds relating the mixing times of an occupancy chain with the underlying dynamics for a general class of models, perhaps building on similar decomposition theorems \cite{MR}.

\bibliographystyle{plainurl}
\bibliography{pontoon.bib}

\appendix
\newcommand{\Kcap}{K}
\newcommand{\Mcap}{M}
\newcommand{\amax}{a^{\text{max}}}
\newcommand{\fcont}{\widehat{f}}
\newcommand{\fopt}{\widehat{f}^*}
\newcommand{\ih}{\mathcal{P}}
\newcommand{\sline}{\mathcal{L}}
\section{Proofs of Lemmas}
\label{app:lemmas}
In this section, we prove two inputs to the proof of Theorem~\ref{thm:zeroorone-full}.

\begin{proof}[Proof of Lemma~\ref{lem:unlikelyconfigurations}]
We have
\begin{align*}
\pi(\OmegaBad_\ell)
&\leq \sum_{\sigma \in \OmegaBad_\ell} \frac{\pi(\sigma)}{\pi(\tau)}\\
&\leq \exp\left\{(\ell + 2w) \log 2 + \beta(H(\tau) - H(\sigma))\right\} \\
&\leq \exp\left\{\ell \log 2 + 2w \log 2 - \beta\epsilon h - \beta \delta \ell  + O_h(1) \right\} \\
&\leq \exp\left\{\ell (\log 2 - \beta\delta) + \left(2\alpha \log 2 - \beta\epsilon\right) h  + O_h(1) \right\}.
\end{align*}
By assumption, $\beta$ satisfies $\log 2 - \beta \delta < 0$, which gives us the following upper bound on the probability of sampling a configuration from $\OmegaBad$:
\begin{align*}
\pi(\OmegaBad) &= \sum_{\ell < \ceil{1/\delta}} \pi(\OmegaBad_\ell) + \sum_{\ell \geq \ceil{1/\delta}} \pi(\OmegaBad_\ell) 
\leq \left(\frac{1}{1-e^{\log 2-\beta\delta}}\right)e^{\left(2\alpha \log 2 - \beta\epsilon\right) h},
\end{align*}
which is exponentially small with $h$ as, by assumption, $\beta$ satisfies $2\alpha \log 2 - \beta\epsilon < 0$.
\end{proof}

\begin{lemma}[Approximation by convex function]\label{lem:convexlemma}
Let $h \geq 1$ be an integer and
suppose $f: [h] \to \mathbb{R}_{\geq 0}$ with $f(1) = 0$. Denoting $\Delta f(k) = f(k+1) - f(k)$ for $k \in [h-1]$, suppose that $0 \leq \Delta f(1) \leq \Delta f(2) \leq \ldots \leq \Delta f(k-1)$.
Then there exists a convex function $\fcont: [0,h] \to \mathbb{R}_{\geq 0}$
such that for all $K_1, K_2 \in [h] = \{1,2,\ldots,h\}$, $K_1 \leq K_2$, we have:
\begin{align*}
- \frac{1}{8}\Delta f(K_1-1)
\leq \int_{K_1-1}^{K_2} \fcont(x) dx - \sum_{K_1}^{K_2}f(k)
\leq \frac{1}{8}\Delta f(K_2),
\end{align*}
where we replace the lower bound with $0$ if $K_1 = 1$, and the upper bound with $0$ if $K_2 = h$.
\end{lemma}

\begin{proof}
The functions we construct will be combinations of straight lines.  For $a \geq 0$, denote by $\sline_{1}$ the constant function $x \mapsto 0$, and for $\Kcap \in \{2,3,\ldots,h\}$, and denote by $\sline_{\Kcap}$ the function $x \mapsto f(\Kcap) + \left(x-\Kcap+\frac{1}{2}\right)\Delta f(\Kcap-1)$.
Using these straight lines, we define the function $\fopt_{0} : [0,h] \to \mathbb{R}_{\geq 0}$, $\fopt_{0}(x) = 0$ for all $x \in [0,h]$.
Then for each $\Kcap \in [h]$, we construct the functions $f_{\Kcap}: [0,h] \times \mathbb{R}_{\geq 0} \to \mathbb{R}$, where: 
\begin{align*}
\fcont_{\Kcap}(x,a) = \begin{cases}
\max\left\{\fopt_{\Kcap-1}(x), \sline_{\Kcap}(x)-a\right\} &\text{ for } x \in [0,\Kcap] \\
0 &\text{ for } x \in (\Kcap,h]
\end{cases}
\end{align*}
where for each $\Kcap \in [h]$, $\fopt_{\Kcap} : [0,h] \to \mathbb{R}$ is defined as $\fopt_{\Kcap}(x) = \fcont_{\Kcap}(x,a_{\Kcap})$, where:
\begin{align*}
a_{\Kcap} = \sup\left\{a \geq 0 : \int_{0}^{\Kcap} \fcont_{\Kcap}(x,a) dx - \sum_{k=1}^{\Kcap}f(k) \geq 0\right\}.
\end{align*}
For this definition to make sense, the set $\left\{a \geq 0 : \int_{0}^{\Kcap} \fcont_{\Kcap}(x,a) dx - \sum_{k=1}^{\Kcap}f(k) \geq 0\right\}$ needs to be bounded and non-empty. This, along with the following properties of these functions, will be shown by induction below. 
%
The inductive hypothesis for $\Kcap \in \{1,2,\ldots,h\}$ states that:
\begin{enumerate}
\item For each $\Mcap \in [\Kcap]$, we have $\displaystyle 0 \leq \int_{0}^{\Mcap} \fcont_{\Kcap}(x,0) dx - \sum_{k=1}^{\Mcap}f(k) \leq \frac{1}{8}\Delta f(\Mcap).$
\item The value $a_{\Kcap} \geq 0$ exists, and the function $\fopt_{\Kcap}$ is 
of the form $\max_{k \in [\Kcap]}\{\sline_k-a_k\}$ on $[0,\Kcap]$ and $0$ outside of $[0,\Kcap]$ (which implies that $\fopt_{\Kcap}$ is convex on $[0,\Kcap]$).
\item For each $\Mcap \in [\Kcap-1]$, we have $\displaystyle 0 \leq \int_{0}^{\Mcap} \fopt_{\Kcap}(x) dx - \sum_{k=1}^{\Mcap}f(k) \leq \frac{1}{8}\Delta f(\Mcap),$ and for the case where $\Mcap = \Kcap$, we have $\displaystyle \int_{0}^{\Kcap} \fopt_{\Kcap}(x) dx - \sum_{k=1}^{\Kcap}f(k) = 0.$
\end{enumerate}
For clarity we denote these three statements by $\ih_1^{\Kcap}$, $\ih_2^{\Kcap}$, $\ih_3^{\Kcap}$ for each $\Kcap \in [h]$ in our proof. The statements $\ih_2^{h}$ and $\ih_3^{h}$ imply the lemma, by using the function $\fopt_h$ for $\fcont$.
For the base case where $K=1$, we have $\fcont_1(\cdot,0) = \sline_1 = 0 = f(1)$, so $\ih_1^1$ clearly holds. For $\ih_2^1$, we have $a_1 = 0$, so $\fopt_1 = f_1(\cdot,0) = \sline_1 = 0$. This also directly implies $\ih_3^1$.

To show $\ih_1^{\Kcap}$ for $\Mcap = \Kcap$, we can reduce the problem to the case where $\Mcap = \Kcap-1$ by noting that
\begin{align*}
&\int_{0}^{\Kcap}\fcont_{\Kcap}(x,0)dx - \sum_{k=1}^{\Kcap}f(k)\\
&= \left(\int_{0}^{\Kcap-1}\fcont_{\Kcap}(x,0)dx - \sum_{k=1}^{\Kcap-1}f(k)\right) + \left(\int_{\Kcap-1}^{\Kcap}\sline_{\Kcap}(x)dx - f(\Kcap)\right),
\end{align*}
where the second term is equal to $0$ by the definition of $\sline_{\Kcap}$.\\

To show the lower bound in $\ih_1^{\Kcap}$ for $\Mcap \in [\Kcap-1]$, we have
\begin{align*}
\int_{0}^{\Mcap} \fcont_{\Kcap}(x,0)dx - \sum_{k=1}^{\Mcap}f(k)
\geq \int_{0}^{\Mcap} \fopt_{\Kcap-1}(x,0)dx - \sum_{k=1}^{\Mcap}f(k) \geq 0,
\end{align*}
where the second inequality applies by $\ih_3^{\Kcap-1}$ if $\Kcap \geq 2$, and is clear for the case where $\Kcap = 1$ as $\fopt_0 = 0$.\\

To show the upper bound in $\ih_1^{\Kcap}$ for $\Mcap \in [\Kcap-1]$, we have two cases, when $\sline_{\Kcap}(\Mcap-1) > \fopt_{\Kcap-1}(\Mcap-1)$ and when $\sline_{\Kcap}(\Mcap-1) \leq \fopt_{\Kcap-1}(\Mcap-1)$.
In the former case, by $\ih_2^{\Kcap-1}$, this implies that $\fcont_{\Kcap}(x,0) = \sline_{\Kcap}(x)$ for all $x \in [\Mcap-1,\Mcap]$, so 
\[\int_{\Mcap-1}^{\Mcap}\fcont_{\Kcap}(x,0)dx = f(\Kcap) + (\Mcap-\Kcap)\Delta f(\Kcap-1) \leq f(\Mcap) = \int_{\Mcap-1}^{\Mcap}\fcont_{\Mcap}(x,0)dx.\]
In addition, for all $x \in [0,\Mcap-1]$, as $\fcont_{\Mcap}(x,0)$ is convex and consists only of lines of gradient at most $\Delta f(\Mcap-1) \leq \Delta f(\Kcap-1)$, we must have $\fcont_{\Kcap}(x,0) \leq \fcont_{\Mcap}(x,0)$. This, along with $\ih_1^{\Mcap}$ (note that $\Mcap < \Kcap$), means that
\begin{align*}
\int_{0}^{\Mcap}\fcont_{\Kcap}(x,0)dx \leq \int_{0}^{\Mcap}\fcont_{\Mcap}(x,0)dx \leq \sum_{k=1}^{\Mcap}f(k) + \frac{1}{8}\Delta f(\Mcap).
\end{align*}
In the latter case, we make use of the following facts:
\begin{itemize}
\item For all $x \in \left[\Mcap-1,\Mcap-\frac{1}{2}\right]$  we have $\fcont_{\Kcap}(x,0) \leq \sline_{\Kcap-1}(x)$.
\item For all $x \in \left[\Mcap-\frac{1}{2},\Mcap\right]$  we have $\fcont_{\Kcap}(x,0) \leq \sline_{\Kcap}(x)$.
\end{itemize}
to state that
\[\int_{\Mcap-1}^{\Mcap} \fcont_{\Kcap}(y,0)dy \leq f(\Mcap) + \frac{1}{8}\Delta f(\Mcap) - \frac{1}{8}\Delta f(\Mcap-1).\]

As $\fcont_{\Kcap}(\Mcap-1,0) \leq \fopt_{\Kcap-1}(\Mcap-1)$, and as
by $\ih_2^{\Kcap-1}$, $\fopt_{\Kcap-1}(x)$
is convex and consists only of lines of gradient at most $\Delta f(\Kcap-2) \leq \Delta f(\Kcap-1)$, we must have 
$\sline_K(x) \leq \fopt_{K-1}(x)$, and thus $\fcont_{\Kcap}(x,0) = \fopt_{K-1}(x)$,
 for all $x \in [0,\Mcap-1]$. Thus, as $M-1 \in [\Kcap-2]$, we have by $\ih_3^{\Kcap-1}$:
\begin{align*}
\int_{0}^{\Mcap}\fcont_{\Kcap}(x,0)dx
&= \int_{0}^{\Mcap-1}\fcont_{\Kcap}(x,0)dx + \int_{\Mcap-1}^{\Mcap}\fcont_{\Kcap}(x,0)dx \\
&\leq \sum_{k=1}^{\Mcap-1}f(k) + \frac{1}{8}\Delta f(\Mcap-1) + f(\Mcap) + \frac{1}{8}\Delta f(\Mcap) - \frac{1}{8}\Delta f(\Mcap-1) \\
&= \sum_{k=1}^{\Mcap}f(k) + \frac{1}{8}\Delta f(\Mcap).
\end{align*}

To show $\ih_2^{\Kcap}$, as we know by $\ih_2^{\Kcap-1}$ that $\fopt_{\Kcap-1}$ is made out of lines of gradient at most $\Delta f(\Kcap-1)$, and by how $\fcont_{\Kcap-1}$ is constructed, we have $\fopt_{\Kcap-1}(\Kcap-1) = \sline_{\Kcap-1}(\Kcap-1)-a_{\Kcap-1}$, which means that $\fcont_{\Kcap}(\cdot,a)$ is equal to $\max\left\{\sline_{\Kcap}-a, \max_{k \in [\Kcap-1]}\{\sline_k-a_k\}\right\}$ whenever $\sline_{\Kcap}(\Kcap-1) - a \geq \sline_{\Kcap1}(\Kcap-1) - a_{\Kcap-1}$, which is if and only if
$0 \leq a \leq \amax_{\Kcap}$, where
\[\amax_{\Kcap} = a_{\Kcap-1} + \frac{1}{2}\left(\Delta f(\Kcap-1) - \Delta f(\Kcap-2)\right).\]
Also, if we set $a = \amax_{\Kcap}$, by $\ih_3^{\Kcap-1}$, we have
\begingroup
\allowdisplaybreaks
\begin{align*}
&\int_{0}^{\Kcap}\fcont_{\Kcap}(x,\amax_{\Kcap})dx - \sum_{k=1}^{\Kcap}f(k)\\
&= \int_{0}^{\Kcap-1}\fcont_{\Kcap-1}(x,\amax_{\Kcap})dx - \sum_{k=1}^{\Kcap-1}f(k) + \int_{\Kcap-1}^{\Kcap}\left(\sline_{\Kcap}(x) - \amax_{\Kcap}\right)dx - f(\Kcap) \\
&\leq \left(\int_{0}^{\Kcap-1}\fopt_{\Kcap-1}(x)dx - \sum_{k=1}^{\Kcap-1}f(k)\right) - \amax_{\Kcap} 
= 0-\amax_{\Kcap} \leq 0.
\end{align*}
\endgroup
This implies that the set $\left\{a \geq 0 : \int_{0}^{\Kcap} \fcont_{\Kcap}(x,a) dx - \sum_{k=1}^{\Kcap}f(k) \geq 0\right\}$ is bounded and non-empty, so $a_{\Kcap}$ is well-defined and not greater than $\amax_{\Kcap}$, which shows $\ih_2^{\Kcap}$.

Finally, for $\ih_3^{\Kcap}$, in the case of $M \in [K-1]$, the lower bound comes from the definition of $\fopt_{\Kcap}$, while the upper bound comes from the upper bound of $\ih_1^{\Kcap}$ and the fact that $\fopt_{\Kcap}(x) \leq \fcont_{\Kcap}(x,0)$ for all $x \in [0,h]$.
For the case of $M=K$, the monotone convergence theorem tells us that the function $a \mapsto \int \fcont_{\Kcap}(x,a)dx$ is continuous on $[0,\amax_{\Kcap}]$. As it is monotone, at least $\sum_{k=1}^{\Kcap}f(k)$ at $0$ and at most $\sum_{k=1}^{\Kcap}f(k)$ at $\amax_{\Kcap}$, we have
\begin{align*}
&\int_{0}^{\Kcap} \fopt_{\Kcap}(x) dx = \int_{0}^{\Kcap} \fcont_{\Kcap}(x,a_{\Kcap}) dx = \sum_{k=1}^{\Kcap}f(k). \qedhere
\end{align*}
\end{proof}

\section{Additional details of the technical lemmas}\label{app:proof}

In this appendix, we further develop the properties and applications of layer sequences, introduced in Section~\ref{sec:layerseq}. First, in Section~\ref{sec:app1}, we collect the proofs of some inputs to the proof of Theorem~\ref{thm:singlelongbridge}. Then, in Section~\ref{sec:app2}, we discuss the main ideas behind the inputs to the proof of Theorem~\ref{thm:zeroorone-full}.

\subsection{Proofs of further inputs to Theorem~\ref{thm:singlelongbridge}}\label{sec:app1}

\begin{proof}[Proof of Lemma~\ref{lem:boundarylengthincrement}]
For $\sigma \in \OmegaExt$, we look at the change in total boundary length when we add layer $M+1$ of $\sigma$ to $\sigma_M$. Let $C_{M+1}$ be the number of components in the restriction of $\sigma$ to only layer $M+1$.

As layer $M+1$ is not completely filled, there are exactly $n_{M+1}-C_{M+1}$ edges between occupied sites within layer $M+1$. Thus, the number of edges of $\Lambda$ with exactly one endpoint at an occupied site in layer $M+1$ is equal to $6n_{M+1} - 2(n_{M+1}-C_{M+1}) = 4n_{M+1} + 2C_{M+1}$. Some of the edges are shared with occupied sites on layer $M$.
Denoting by $E_{M+1}$ the number of edges between layers $M$ and $M+1$, we can see that adding layer $M+1$ to $\sigma_M$ adds $4n_{M+1} + 2C_{M+1} - 2E_{M+1}$ to the boundary length of $\sigma_{M+1}$.

We compute an upper bound for $E_{M+1}$ by counting the number of occupied sites in layer $M+1$ {with at least one and with at least two edges} to layer $M$ respectively.
Layer $M$, with $C_M$ components, provides $n_M + C_M$ occupied sites in layer $M+1$ with at least one edge to layer $M$, and $n_M - C_M$ occupied sites in layer $M+1$ with two edges to layer $M$. A site in layer $M+1$ cannot have more than two edges to layer $M$. This gives an upper bound of $\min\{n_M+c_M,n_{M+1}\}$ for the former, and $\min\{n_M-c_M,n_{M+1}\}$ for the latter.
This gives us the following lower bound:
\begin{align*}
B(\sigma_{M+1}) &- B(\sigma_M) \\
\geq \ & 4n_{M+1} + 2C_{M+1} - 2\min\{n_M+C_M,n_{M+1}\} - 2\min\{n_M-C_M,n_{M+1}\} \\
\geq \ & 2C_{M+1} + 2\max\{n_{M+1}-n_M-C_M,0\} + 2\max\{n_{M+1}-n_M+C_M,0\} \\
\geq \ & 2 + 2\max\{n_{M+1}-n_M-1,0\} + 2\max\{n_{M+1}-n_M+1,0\}.
\end{align*}
The last inequality holds because $C_{M+1} \geq 1$ and $\max\{n_{M+1}-n_M-C_M,0\} + \max\{n_{M+1}-n_M+C_M,0\} \geq \max\{n_{M+1}-n_M-1,0\} + \max\{n_{M+1}-n_M+1,0\}$ for any value of $n_{M+1}-n_M$, as long as $C_M \geq 1$ (which is true as $\sigma \in \OmegaExt$).
To show the second statement of the lemma, we can construct $\sigma$ layer by layer, where each layer from $\max\{1,\rbot\}$ to $\rtop$ has exactly one component and $M_{M+1}$ is maximized for each value of $M$.
\end{proof}

\begin{proof}[Proof of Lemma~\ref{lem:epsboundary}]
In this proof, we will use $(x)_+$ to denote $\max\{x,0\}$ for $x \in \mathbb{R}$.
Fix $\eps \geq 1/h$ and a configuration $\sigma \in \MB_\eps$. By definition \eqref{eqn:mbdef}, there is an $a \in (0,1)$ such that $\sigma$ has multiple $(a,a+\eps)$-bridges. Set $M = \floor*{a h}$, and let $\RegionH$ represent the set of sites on layers $M$ to $h$. We can divide the elements of $\sigma$ within $\RegionH$ into connected components $V_1, V_2, \ldots, V_P$ for some integer $P$. Denote by $\Vbase$ the elements of $\sigma$ not in $\RegionH$.
For each $j \in \{1,2,\ldots,P\}$, we denote $\configpart{j} = V_j \cup \Vbase$
which we note will be a configuration in $\OmegaExt$,
and let $\seqpart{j} = (\npart{j}_k)_{k \in \{0,1,\ldots,h\}}$ be its layer sequence.

By applying Lemma~\ref{lem:boundarylengthincrement}, the increase in boundary length when adding the sites $V_j$ to the restriction of $\sigma$ to the layers $1,2,\ldots,M-1$ must be at least $\Bhatpart{j}$, where
\begin{align*}
\Bhatpart{j}
&= \sum_{k=M}^{\rtoppart{j}-1}\left(2 + 2\left(\npart{j}_{k+1}-\npart{j}_k+1\right)_+ + 2\left(\npart{j}_{k+1}-\npart{j}_k-1\right)_+\right),
\end{align*}
and where $\rtoppart{j}$ denotes the topmost non-empty layer of $\configpart{j}$. We can thus give the following lower bound for the boundary length of $\sigma$:
\[B(\sigma) \geq \Bhat_M(\layerseq) + \sum_{j=1}^P \Bhatpart{j}.\]
Denoting by $\layerseq = (n_k)_{k \in \{0,1,\ldots,h\}}$ the layer sequence of $\sigma$ itself, by Lemma~\ref{lem:boundarylengthlowerbound}, there exists a configuration $\tau$ with layer sequence $\layerseq$ where:
\[B(\tau) = \Bhat_M(\layerseq) + \sum_{k=M}^{\rtop-1}\left(2 + 2\left(n_{k+1}-n_k+1\right)_+ + 2\left(n_{k+1}-n_k-1\right)_+ \right).\]
Without loss of generality, we may assume that $\rtop = \rtoppart{1} \geq \rtoppart{2} \geq \ldots \geq \rtoppart{P}$. 
\begingroup
\allowdisplaybreaks
\begin{align*}
&B(\sigma)  - B(\tau)
\geq \sum_{j=1}^P \Bhatpart{j} - \sum_{k=M}^{\rtop}\left(2 + 2\left(n_{k+1}-n_k+1\right)_+ + 2\left(n_{k+1}-n_k-1\right)_+\right) \\
&= -\sum_{j=2}^P\sum_{k=M}^{\rtoppart{j}}2 +
2\sum_{k=M}^{\rtoppart{j}}\left(\sum_{\substack{j \in [P] \\ k \leq \rtoppart{j}-1}}\!\left(\npart{j}_{k+1}-\npart{j}_k+\!1\right)_+ - 
\left(n_{k+1}-n_k+1\right)_+\right) \\
&\hspace{0.85in} + 2\sum_{k=M}^{\rtoppart{j}}\left(\sum_{\substack{j \in [P] \\ k \leq \rtoppart{j}-1}}\left(\npart{j}_{k+1}-\npart{j}_k-1\right)_+ - \left(n_{k+1}-n_k-1\right)_+
\right).
\end{align*}
\endgroup
The second and third terms of the above expression are non-negative as for $I \in \{1,-1\}$, we have
\begingroup
\allowdisplaybreaks
\begin{align*}
\sum_{\substack{j \in [P] \\ k \leq \rtoppart{j}-1}} & \left(\npart{j}_{k+1}-\npart{j}_k+I\right)_+ - 
\left(n_{k+1}-n_k+I\right)_+ \\ 
&= \sum_{\substack{j \in [P] \\ k \leq \rtoppart{j}}}\left(\npart{j}_{k+1}-\npart{j}_k+I\right)_+ - 
\Big(\sum_{\substack{j \in [P] \\ k \leq \rtoppart{j}}}(n_{k+1}-n_k)+I\Big)_+ \geq 0.
\end{align*}
\endgroup
This allows us to conclude that
\begin{align}
\label{p6_expression}
B(\sigma) - B(\tau)
\geq \sum_{j=2}^P\sum_{k=M}^{\rtoppart{j}}2
= 2\sum_{j=2}^P (\rtoppart{j} - M + 1).
\end{align}
If $\sigma$ has multiple $(a,a+\eps)$-bridges, we must have $P \geq 2$ and
$\rtoppart{2} - M + 1\geq \epsilon h$,
and so according to~(\ref{p6_expression}) we must have $B(\sigma) \geq B(\tau) + 2\epsilon h$.
\end{proof}

\subsection{Proofs of further inputs to Theorem~\ref{thm:zeroorone-full}}\label{sec:app2}
These lemmas arise from the analysis of a transformation of layer sequences which gives a significant improvement (more negative) in the Hamiltonian when a configuration without a bridge to the food is extended to reach the food, so long as the scent parameter $\eta$ is sufficiently large. 

Let $\sigma$ be an arbitrary configuration with layer sequence $\layerseq = (n_k)_{k \in [h]}$. Suppose that $\layerseq \in \LayerSet^{0,\toplayer}$, where $\toplayer \leq h-1$ is the highest occupied layer in $\sigma$.
We will transform $\layerseq$ into a new layer sequence $\layerseq^{\spost}$ in $\LayerSet^{0,h}$ (i.e. it reaches $\Lambda_h$).
We then show that with a sufficiently large value of $\eta$ (or $\varphi$), the increase in scent intensity going from $\layerseq$ to $\layerseq^{\spost}$ will more than compensate for the potential increase in boundary length.
However, the transformation from $\layerseq$ to $\layerseq^{\spost}$ is complex, and we will thus split it into multiple steps---pre-process, the main transformation, and post-process. The layer sequences after each of these steps are denoted $\layerseq^{\spre}$, $\layerseq^{\stran}$ and $\layerseq^{\spost}$ respectively.

\begin{figure*}[h]
\begin{subfigure}{.48\textwidth}
  \begin{center}
  \begin{tikzpicture}[x=0.6cm,y=0.6cm]
  \draw[lightgray] (5.48483,-0.5) -- (7.50555,-4);
\draw[lightgray] (0,-5) -- (2.59808,-0.5);
\draw[lightgray] (0,-5) -- (6.9282,-5);
\draw[lightgray] (1.73205,-2) -- (3.4641,-5);
\draw[lightgray] (1.73205,-2) -- (8.66025,-2);
\draw[lightgray] (2.02073,-1.5) -- (4.04145,-5);
\draw[lightgray] (2.02073,-1.5) -- (8.94893,-1.5);
\draw[lightgray] (6.06218,-0.5) -- (7.79423,-3.5);
\draw[lightgray] (0.57735,-4) -- (1.1547,-5);
\draw[lightgray] (0.57735,-4) -- (7.50555,-4);
\draw[lightgray] (3.17543,-0.5) -- (5.7735,-5);
\draw[lightgray] (4.33013,-0.5) -- (6.9282,-5);
\draw[lightgray] (5.19615,-5) -- (7.79423,-0.5);
\draw[lightgray] (6.35085,-5) -- (8.94893,-0.5);
\draw[lightgray] (4.90748,-0.5) -- (7.21688,-4.5);
\draw[lightgray] (6.9282,-5) -- (9.52628,-0.5);
\draw[lightgray] (0.866025,-3.5) -- (1.73205,-5);
\draw[lightgray] (0.866025,-3.5) -- (7.79423,-3.5);
\draw[lightgray] (4.04145,-5) -- (6.63953,-0.5);
\draw[lightgray] (1.44338,-2.5) -- (2.88675,-5);
\draw[lightgray] (1.44338,-2.5) -- (8.37158,-2.5);
\draw[lightgray] (2.59808,-0.5) -- (5.19615,-5);
\draw[lightgray] (2.59808,-0.5) -- (9.52628,-0.5);
\draw[lightgray] (3.75278,-0.5) -- (6.35085,-5);
\draw[lightgray] (4.6188,-5) -- (7.21688,-0.5);
\draw[lightgray] (5.7735,-5) -- (8.37158,-0.5);
\draw[lightgray] (2.88675,-5) -- (5.48483,-0.5);
\draw[lightgray] (8.94893,-0.5) -- (9.2376,-1);
\draw[lightgray] (3.4641,-5) -- (6.06218,-0.5);
\draw[lightgray] (2.3094,-1) -- (4.6188,-5);
\draw[lightgray] (2.3094,-1) -- (9.2376,-1);
\draw[lightgray] (0.288675,-4.5) -- (0.57735,-5);
\draw[lightgray] (0.288675,-4.5) -- (7.21688,-4.5);
\draw[lightgray] (1.1547,-3) -- (2.3094,-5);
\draw[lightgray] (1.1547,-3) -- (8.0829,-3);
\draw[lightgray] (0.57735,-5) -- (3.17543,-0.5);
\draw[lightgray] (1.73205,-5) -- (4.33013,-0.5);
\draw[lightgray] (6.63953,-0.5) -- (8.0829,-3);
\draw[lightgray] (1.1547,-5) -- (3.75278,-0.5);
\draw[lightgray] (7.79423,-0.5) -- (8.66025,-2);
\draw[lightgray] (2.3094,-5) -- (4.90748,-0.5);
\draw[lightgray] (7.21688,-0.5) -- (8.37158,-2.5);
\draw[lightgray] (8.37158,-0.5) -- (8.94893,-1.5);
\draw[black, line width=0.2mm, fill=white] (0,-5) circle (0.216);
\draw[black, line width=0.2mm, fill=white] (0.288675,-4.5) circle (0.216);
\draw[black, line width=0.2mm, fill=white] (0.57735,-4) circle (0.216);
\draw[black, line width=0.2mm, fill=white] (0.57735,-5) circle (0.216);
\draw[black, line width=0.2mm, fill=white] (0.866025,-4.5) circle (0.216);
\draw[black, line width=0.2mm, fill=white] (1.1547,-4) circle (0.216);
\draw[black, line width=0.2mm, fill=white] (1.1547,-5) circle (0.216);
\draw[black, line width=0.2mm, fill=white] (1.44338,-4.5) circle (0.216);
\draw[black, line width=0.2mm, fill=white] (1.73205,-5) circle (0.216);
\draw[black, line width=0.2mm, fill=white] (2.02073,-4.5) circle (0.216);
\draw[black, line width=0.2mm, fill=white] (2.3094,-4) circle (0.216);
\draw[black, line width=0.2mm, fill=white] (2.3094,-5) circle (0.216);
\draw[black, line width=0.2mm, fill=white] (2.59808,-2.5) circle (0.216);
\draw[black, line width=0.2mm, fill=white] (2.59808,-4.5) circle (0.216);
\draw[black, line width=0.2mm, fill=white] (2.88675,-2) circle (0.216);
\draw[black, line width=0.2mm, fill=white] (2.88675,-3) circle (0.216);
\draw[black, line width=0.2mm, fill=white] (2.88675,-5) circle (0.216);
\draw[black, line width=0.2mm, fill=white] (3.17543,-4.5) circle (0.216);
\draw[black, line width=0.2mm, fill=white] (3.4641,-2) circle (0.216);
\draw[black, line width=0.2mm, fill=white] (3.4641,-4) circle (0.216);
\draw[black, line width=0.2mm, fill=white] (3.4641,-5) circle (0.216);
\draw[black, line width=0.2mm, fill=white] (3.75278,-4.5) circle (0.216);
\draw[black, line width=0.2mm, fill=white] (4.04145,-2) circle (0.216);
\draw[black, line width=0.2mm, fill=white] (4.04145,-4) circle (0.216);
\draw[black, line width=0.2mm, fill=white] (4.04145,-5) circle (0.216);
\draw[black, line width=0.2mm, fill=white] (4.33013,-2.5) circle (0.216);
\draw[black, line width=0.2mm, fill=white] (4.33013,-3.5) circle (0.216);
\draw[black, line width=0.2mm, fill=white] (4.33013,-4.5) circle (0.216);
\draw[black, line width=0.2mm, fill=white] (4.6188,-2) circle (0.216);
\draw[black, line width=0.2mm, fill=white] (4.6188,-3) circle (0.216);
\draw[black, line width=0.2mm, fill=white] (4.6188,-4) circle (0.216);
\draw[black, line width=0.2mm, fill=white] (4.90748,-2.5) circle (0.216);
\draw[black, line width=0.2mm, fill=white] (4.90748,-3.5) circle (0.216);
\draw[black, line width=0.2mm, fill=white] (5.19615,-3) circle (0.216);
\draw[black, line width=0.2mm, fill=white] (5.19615,-4) circle (0.216);
\draw[black, line width=0.2mm, fill=white] (5.19615,-5) circle (0.216);
\draw[black, line width=0.2mm, fill=white] (5.48483,-2.5) circle (0.216);
\draw[black, line width=0.2mm, fill=white] (5.48483,-3.5) circle (0.216);
\draw[black, line width=0.2mm, fill=white] (5.7735,-2) circle (0.216);
\draw[black, line width=0.2mm, fill=white] (5.7735,-3) circle (0.216);
\draw[black, line width=0.2mm, fill=white] (5.7735,-4) circle (0.216);
\draw[black, line width=0.2mm, fill=white] (5.7735,-5) circle (0.216);
\draw[black, line width=0.2mm, fill=white] (6.06218,-2.5) circle (0.216);
\draw[black, line width=0.2mm, fill=white] (6.06218,-3.5) circle (0.216);
\draw[black, line width=0.2mm, fill=white] (6.35085,-2) circle (0.216);
\draw[black, line width=0.2mm, fill=white] (6.35085,-4) circle (0.216);
\draw[black, line width=0.2mm, fill=white] (6.35085,-5) circle (0.216);
\draw[black, line width=0.2mm, fill=white] (6.63953,-2.5) circle (0.216);
\draw[black, line width=0.2mm, fill=white] (6.9282,-2) circle (0.216);
\draw[black, line width=0.2mm, fill=white] (6.9282,-5) circle (0.216);
\draw[black, line width=0.2mm, fill=white] (7.21688,-1.5) circle (0.216);
\draw[black, line width=0.2mm, fill=white] (7.21688,-4.5) circle (0.216);
\draw[black, line width=0.2mm, fill=white] (7.50555,-4) circle (0.216);
\draw[black, line width=0.2mm, fill=white] (7.79423,-1.5) circle (0.216);
\draw[black, line width=0.2mm, fill=white] (8.0829,-2) circle (0.216);
\draw[black, line width=0.2mm, fill=white] (8.66025,-2) circle (0.216);
  \end{tikzpicture}
  \end{center}
  \vskip.02in
  \caption{Original configuration $\sigma$ with some number of unused particles (not shown).}
  \label{fig:transform1}
\end{subfigure}~~%
\begin{subfigure}{.48\textwidth}
  \begin{center}
  \begin{tikzpicture}[x=0.6cm,y=0.6cm]
  \draw[lightgray] (5.48483,-0.5) -- (7.50555,-4);
\draw[lightgray] (0,-5) -- (2.59808,-0.5);
\draw[lightgray] (0,-5) -- (6.9282,-5);
\draw[lightgray] (1.73205,-2) -- (3.4641,-5);
\draw[lightgray] (1.73205,-2) -- (8.66025,-2);
\draw[lightgray] (2.02073,-1.5) -- (4.04145,-5);
\draw[lightgray] (2.02073,-1.5) -- (8.94893,-1.5);
\draw[lightgray] (6.06218,-0.5) -- (7.79423,-3.5);
\draw[lightgray] (0.57735,-4) -- (1.1547,-5);
\draw[lightgray] (0.57735,-4) -- (7.50555,-4);
\draw[lightgray] (3.17543,-0.5) -- (5.7735,-5);
\draw[lightgray] (4.33013,-0.5) -- (6.9282,-5);
\draw[lightgray] (5.19615,-5) -- (7.79423,-0.5);
\draw[lightgray] (6.35085,-5) -- (8.94893,-0.5);
\draw[lightgray] (4.90748,-0.5) -- (7.21688,-4.5);
\draw[lightgray] (6.9282,-5) -- (9.52628,-0.5);
\draw[lightgray] (0.866025,-3.5) -- (1.73205,-5);
\draw[lightgray] (0.866025,-3.5) -- (7.79423,-3.5);
\draw[lightgray] (4.04145,-5) -- (6.63953,-0.5);
\draw[lightgray] (1.44338,-2.5) -- (2.88675,-5);
\draw[lightgray] (1.44338,-2.5) -- (8.37158,-2.5);
\draw[lightgray] (2.59808,-0.5) -- (5.19615,-5);
\draw[lightgray] (2.59808,-0.5) -- (9.52628,-0.5);
\draw[lightgray] (3.75278,-0.5) -- (6.35085,-5);
\draw[lightgray] (4.6188,-5) -- (7.21688,-0.5);
\draw[lightgray] (5.7735,-5) -- (8.37158,-0.5);
\draw[lightgray] (2.88675,-5) -- (5.48483,-0.5);
\draw[lightgray] (8.94893,-0.5) -- (9.2376,-1);
\draw[lightgray] (3.4641,-5) -- (6.06218,-0.5);
\draw[lightgray] (2.3094,-1) -- (4.6188,-5);
\draw[lightgray] (2.3094,-1) -- (9.2376,-1);
\draw[lightgray] (0.288675,-4.5) -- (0.57735,-5);
\draw[lightgray] (0.288675,-4.5) -- (7.21688,-4.5);
\draw[lightgray] (1.1547,-3) -- (2.3094,-5);
\draw[lightgray] (1.1547,-3) -- (8.0829,-3);
\draw[lightgray] (0.57735,-5) -- (3.17543,-0.5);
\draw[lightgray] (1.73205,-5) -- (4.33013,-0.5);
\draw[lightgray] (6.63953,-0.5) -- (8.0829,-3);
\draw[lightgray] (1.1547,-5) -- (3.75278,-0.5);
\draw[lightgray] (7.79423,-0.5) -- (8.66025,-2);
\draw[lightgray] (2.3094,-5) -- (4.90748,-0.5);
\draw[lightgray] (7.21688,-0.5) -- (8.37158,-2.5);
\draw[lightgray] (8.37158,-0.5) -- (8.94893,-1.5);
\draw[black, line width=0.2mm, fill=white] (0.57735,-5) circle (0.216);
\draw[black, line width=0.2mm, fill=white] (1.1547,-5) circle (0.216);
\draw[black, line width=0.2mm, fill=white] (1.73205,-5) circle (0.216);
\draw[black, line width=0.2mm, fill=white] (2.3094,-4) circle (0.216);
\draw[black, line width=0.2mm, fill=white] (2.3094,-5) circle (0.216);
\draw[black, line width=0.2mm, fill=white] (2.59808,-4.5) circle (0.216);
\draw[black, line width=0.2mm, fill=white] (2.88675,-4) circle (0.216);
\draw[black, line width=0.2mm, fill=white] (2.88675,-5) circle (0.216);
\draw[black, line width=0.2mm, fill=white] (3.17543,-4.5) circle (0.216);
\draw[black, line width=0.2mm, fill=white] (3.4641,-4) circle (0.216);
\draw[black, line width=0.2mm, fill=white] (3.4641,-5) circle (0.216);
\draw[black, line width=0.2mm, fill=white] (3.75278,-4.5) circle (0.216);
\draw[black, line width=0.2mm, fill=white] (4.04145,-2) circle (0.216);
\draw[black, line width=0.2mm, fill=white] (4.04145,-4) circle (0.216);
\draw[black, line width=0.2mm, fill=white] (4.04145,-5) circle (0.216);
\draw[black, line width=0.2mm, fill=white] (4.33013,-4.5) circle (0.216);
\draw[black, line width=0.2mm, fill=white] (4.6188,-2) circle (0.216);
\draw[black, line width=0.2mm, fill=white] (4.6188,-4) circle (0.216);
\draw[black, line width=0.2mm, fill=white] (4.6188,-5) circle (0.216);
\draw[black, line width=0.2mm, fill=white] (4.90748,-4.5) circle (0.216);
\draw[black, line width=0.2mm, fill=white] (5.19615,-2) circle (0.216);
\draw[black, line width=0.2mm, fill=white] (5.19615,-4) circle (0.216);
\draw[black, line width=0.2mm, fill=white] (5.19615,-5) circle (0.216);
\draw[black, line width=0.2mm, fill=white] (5.48483,-2.5) circle (0.216);
\draw[black, line width=0.2mm, fill=white] (5.48483,-4.5) circle (0.216);
\draw[black, line width=0.2mm, fill=white] (5.7735,-2) circle (0.216);
\draw[black, line width=0.2mm, fill=white] (5.7735,-4) circle (0.216);
\draw[black, line width=0.2mm, fill=white] (5.7735,-5) circle (0.216);
\draw[black, line width=0.2mm, fill=white] (6.06218,-2.5) circle (0.216);
\draw[black, line width=0.2mm, fill=white] (6.06218,-3.5) circle (0.216);
\draw[black, line width=0.2mm, fill=white] (6.06218,-4.5) circle (0.216);
\draw[black, line width=0.2mm, fill=white] (6.35085,-2) circle (0.216);
\draw[black, line width=0.2mm, fill=white] (6.35085,-3) circle (0.216);
\draw[black, line width=0.2mm, fill=white] (6.35085,-4) circle (0.216);
\draw[black, line width=0.2mm, fill=white] (6.35085,-5) circle (0.216);
\draw[black, line width=0.2mm, fill=white] (6.63953,-2.5) circle (0.216);
\draw[black, line width=0.2mm, fill=white] (6.63953,-3.5) circle (0.216);
\draw[black, line width=0.2mm, fill=white] (6.63953,-4.5) circle (0.216);
\draw[black, line width=0.2mm, fill=white] (6.9282,-2) circle (0.216);
\draw[black, line width=0.2mm, fill=white] (6.9282,-3) circle (0.216);
\draw[black, line width=0.2mm, fill=white] (6.9282,-4) circle (0.216);
\draw[black, line width=0.2mm, fill=white] (6.9282,-5) circle (0.216);
\draw[black, line width=0.2mm, fill=white] (7.21688,-2.5) circle (0.216);
\draw[black, line width=0.2mm, fill=white] (7.21688,-3.5) circle (0.216);
\draw[black, line width=0.2mm, fill=white] (7.21688,-4.5) circle (0.216);
\draw[black, line width=0.2mm, fill=white] (7.50555,-2) circle (0.216);
\draw[black, line width=0.2mm, fill=white] (7.50555,-3) circle (0.216);
\draw[black, line width=0.2mm, fill=white] (7.50555,-4) circle (0.216);
\draw[black, line width=0.2mm, fill=white] (7.79423,-2.5) circle (0.216);
\draw[black, line width=0.2mm, fill=white] (7.79423,-3.5) circle (0.216);
\draw[black, line width=0.2mm, fill=white] (8.0829,-2) circle (0.216);
\draw[black, line width=0.2mm, fill=white] (8.0829,-3) circle (0.216);
\draw[black, line width=0.2mm, fill=white] (8.37158,-1.5) circle (0.216);
\draw[black, line width=0.2mm, fill=white] (8.37158,-2.5) circle (0.216);
\draw[black, line width=0.2mm, fill=white] (8.66025,-2) circle (0.216);
\draw[black, line width=0.2mm, fill=white] (8.94893,-1.5) circle (0.216);
  \end{tikzpicture}
  \end{center}
  \caption{Representation of layer sequence $\layerseq$ (right-justified $\sigma$) with unused particles not shown.}
  \label{fig:transform2}
\end{subfigure}
\begin{subfigure}{.48\textwidth}
  \begin{center}
  \begin{tikzpicture}[x=0.6cm,y=0.6cm]
  \draw[lightgray] (5.48483,-0.5) -- (7.50555,-4);
\draw[lightgray] (0,-5) -- (2.59808,-0.5);
\draw[lightgray] (0,-5) -- (6.9282,-5);
\draw[lightgray] (1.73205,-2) -- (3.4641,-5);
\draw[lightgray] (1.73205,-2) -- (8.66025,-2);
\draw[lightgray] (2.02073,-1.5) -- (4.04145,-5);
\draw[lightgray] (2.02073,-1.5) -- (8.94893,-1.5);
\draw[lightgray] (6.06218,-0.5) -- (7.79423,-3.5);
\draw[lightgray] (0.57735,-4) -- (1.1547,-5);
\draw[lightgray] (0.57735,-4) -- (7.50555,-4);
\draw[lightgray] (3.17543,-0.5) -- (5.7735,-5);
\draw[lightgray] (4.33013,-0.5) -- (6.9282,-5);
\draw[lightgray] (5.19615,-5) -- (7.79423,-0.5);
\draw[lightgray] (6.35085,-5) -- (8.94893,-0.5);
\draw[lightgray] (4.90748,-0.5) -- (7.21688,-4.5);
\draw[lightgray] (6.9282,-5) -- (9.52628,-0.5);
\draw[lightgray] (0.866025,-3.5) -- (1.73205,-5);
\draw[lightgray] (0.866025,-3.5) -- (7.79423,-3.5);
\draw[lightgray] (4.04145,-5) -- (6.63953,-0.5);
\draw[lightgray] (1.44338,-2.5) -- (2.88675,-5);
\draw[lightgray] (1.44338,-2.5) -- (8.37158,-2.5);
\draw[lightgray] (2.59808,-0.5) -- (5.19615,-5);
\draw[lightgray] (2.59808,-0.5) -- (9.52628,-0.5);
\draw[lightgray] (3.75278,-0.5) -- (6.35085,-5);
\draw[lightgray] (4.6188,-5) -- (7.21688,-0.5);
\draw[lightgray] (5.7735,-5) -- (8.37158,-0.5);
\draw[lightgray] (2.88675,-5) -- (5.48483,-0.5);
\draw[lightgray] (8.94893,-0.5) -- (9.2376,-1);
\draw[lightgray] (3.4641,-5) -- (6.06218,-0.5);
\draw[lightgray] (2.3094,-1) -- (4.6188,-5);
\draw[lightgray] (2.3094,-1) -- (9.2376,-1);
\draw[lightgray] (0.288675,-4.5) -- (0.57735,-5);
\draw[lightgray] (0.288675,-4.5) -- (7.21688,-4.5);
\draw[lightgray] (1.1547,-3) -- (2.3094,-5);
\draw[lightgray] (1.1547,-3) -- (8.0829,-3);
\draw[lightgray] (0.57735,-5) -- (3.17543,-0.5);
\draw[lightgray] (1.73205,-5) -- (4.33013,-0.5);
\draw[lightgray] (6.63953,-0.5) -- (8.0829,-3);
\draw[lightgray] (1.1547,-5) -- (3.75278,-0.5);
\draw[lightgray] (7.79423,-0.5) -- (8.66025,-2);
\draw[lightgray] (2.3094,-5) -- (4.90748,-0.5);
\draw[lightgray] (7.21688,-0.5) -- (8.37158,-2.5);
\draw[lightgray] (8.37158,-0.5) -- (8.94893,-1.5);
\draw[black, line width=0.2mm, fill=white] (0.57735,-5) circle (0.252);
\draw[black, line width=0.16mm] (0.57735,-5) circle (0.18);
\draw[black, line width=0.2mm, fill=white] (1.1547,-5) circle (0.252);
\draw[black, line width=0.16mm] (1.1547,-5) circle (0.18);
\draw[black, line width=0.2mm, fill=white] (1.73205,-5) circle (0.216);
\draw[black, line width=0.2mm, fill=white] (2.3094,-4) circle (0.216);
\draw[black, line width=0.2mm, fill=white] (2.3094,-5) circle (0.252);
\draw[black, line width=0.16mm] (2.3094,-5) circle (0.18);
\draw[black, line width=0.2mm, fill=white] (2.59808,-4.5) circle (0.252);
\draw[black, line width=0.16mm] (2.59808,-4.5) circle (0.18);
\draw[black, line width=0.2mm, fill=white] (2.88675,-4) circle (0.252);
\draw[black, line width=0.16mm] (2.88675,-4) circle (0.18);
\draw[black, line width=0.2mm, fill=white] (2.88675,-5) circle (0.216);
\draw[black, line width=0.2mm, fill=white] (3.17543,-4.5) circle (0.216);
\draw[black, line width=0.2mm, fill=white] (3.4641,-4) circle (0.216);
\draw[black, line width=0.2mm, fill=white] (3.4641,-5) circle (0.216);
\draw[black, line width=0.2mm, fill=white] (3.75278,-4.5) circle (0.216);
\draw[black, line width=0.2mm, fill=white] (4.04145,-4) circle (0.216);
\draw[black, line width=0.2mm, fill=white] (4.04145,-5) circle (0.216);
\draw[black, line width=0.2mm, fill=white] (4.33013,-4.5) circle (0.216);
\draw[black, line width=0.2mm, fill=white] (4.6188,-2) circle (0.216);
\draw[black, line width=0.2mm, fill=white] (4.6188,-4) circle (0.216);
\draw[black, line width=0.2mm, fill=white] (4.6188,-5) circle (0.216);
\draw[black, line width=0.2mm, fill=white] (4.90748,-4.5) circle (0.216);
\draw[black, line width=0.2mm, fill=white] (5.19615,-2) circle (0.216);
\draw[black, line width=0.2mm, fill=white] (5.19615,-4) circle (0.216);
\draw[black, line width=0.2mm, fill=white] (5.19615,-5) circle (0.252);
\draw[black, line width=0.16mm] (5.19615,-5) circle (0.18);
\draw[black, line width=0.2mm, fill=white] (5.48483,-2.5) circle (0.216);
\draw[black, line width=0.2mm, fill=white] (5.48483,-4.5) circle (0.252);
\draw[black, line width=0.16mm] (5.48483,-4.5) circle (0.18);
\draw[black, line width=0.2mm, fill=white] (5.7735,-2) circle (0.216);
\draw[black, line width=0.2mm, fill=white] (5.7735,-4) circle (0.252);
\draw[black, line width=0.16mm] (5.7735,-4) circle (0.18);
\draw[black, line width=0.2mm, fill=white] (5.7735,-5) circle (0.252);
\draw[black, line width=0.16mm] (5.7735,-5) circle (0.18);
\draw[black, line width=0.2mm, fill=white] (6.06218,-2.5) circle (0.216);
\draw[black, line width=0.2mm, fill=white] (6.06218,-3.5) circle (0.252);
\draw[black, line width=0.16mm] (6.06218,-3.5) circle (0.18);
\draw[black, line width=0.2mm, fill=white] (6.06218,-4.5) circle (0.252);
\draw[black, line width=0.16mm] (6.06218,-4.5) circle (0.18);
\draw[black, line width=0.2mm, fill=white] (6.35085,-2) circle (0.216);
\draw[black, line width=0.2mm, fill=white] (6.35085,-3) circle (0.252);
\draw[black, line width=0.16mm] (6.35085,-3) circle (0.18);
\draw[black, line width=0.2mm, fill=white] (6.35085,-4) circle (0.252);
\draw[black, line width=0.16mm] (6.35085,-4) circle (0.18);
\draw[black, line width=0.2mm, fill=white] (6.35085,-5) circle (0.252);
\draw[black, line width=0.16mm] (6.35085,-5) circle (0.18);
\draw[black, line width=0.2mm, fill=white] (6.63953,-2.5) circle (0.252);
\draw[black, line width=0.16mm] (6.63953,-2.5) circle (0.18);
\draw[black, line width=0.2mm, fill=white] (6.63953,-3.5) circle (0.252);
\draw[black, line width=0.16mm] (6.63953,-3.5) circle (0.18);
\draw[black, line width=0.2mm, fill=white] (6.63953,-4.5) circle (0.252);
\draw[black, line width=0.16mm] (6.63953,-4.5) circle (0.18);
\draw[black, line width=0.2mm, fill=white] (6.9282,-2) circle (0.252);
\draw[black, line width=0.16mm] (6.9282,-2) circle (0.18);
\draw[black, line width=0.2mm, fill=white] (6.9282,-3) circle (0.252);
\draw[black, line width=0.16mm] (6.9282,-3) circle (0.18);
\draw[black, line width=0.2mm, fill=white] (6.9282,-4) circle (0.252);
\draw[black, line width=0.16mm] (6.9282,-4) circle (0.18);
\draw[fill=black] (6.9282,-5) circle (0.216);
\draw[black, line width=0.2mm, fill=white] (7.21688,-2.5) circle (0.252);
\draw[black, line width=0.16mm] (7.21688,-2.5) circle (0.18);
\draw[black, line width=0.2mm, fill=white] (7.21688,-3.5) circle (0.252);
\draw[black, line width=0.16mm] (7.21688,-3.5) circle (0.18);
\draw[fill=black] (7.21688,-4.5) circle (0.216);
\draw[black, line width=0.2mm, fill=white] (7.50555,-2) circle (0.252);
\draw[black, line width=0.16mm] (7.50555,-2) circle (0.18);
\draw[black, line width=0.2mm, fill=white] (7.50555,-3) circle (0.252);
\draw[black, line width=0.16mm] (7.50555,-3) circle (0.18);
\draw[fill=black] (7.50555,-4) circle (0.216);
\draw[black, line width=0.2mm, fill=white] (7.79423,-2.5) circle (0.252);
\draw[black, line width=0.16mm] (7.79423,-2.5) circle (0.18);
\draw[fill=black] (7.79423,-3.5) circle (0.216);
\draw[black, line width=0.2mm, fill=white] (8.0829,-2) circle (0.252);
\draw[black, line width=0.16mm] (8.0829,-2) circle (0.18);
\draw[fill=black] (8.0829,-3) circle (0.216);
\draw[fill=black] (8.37158,-2.5) circle (0.216);
\draw[fill=black] (8.66025,-2) circle (0.216);
\draw[fill=black] (8.94893,-1.5) circle (0.216);
\draw[fill=black] (9.2376,-1) circle (0.216);
\draw[fill=black] (9.52628,-0.5) circle (0.216);
  \end{tikzpicture}
  \end{center}
  \vskip.08in
  \caption{Representation of layer sequence $\layerseq^{\spre}$. There were no empty columns in $\layerseq$, so all unused particles remain unused, while $2$ particles from the topmost layers are moved to form a single column to the top. Particles from $(\nfull_k)_{k \in [h]}$, $(\nbot_k)_{k \in [h]}$ and $(\ntop_k)_{k \in [h]}$ represented by black circles, double circles and single circles respectively.}
  \label{fig:transform3}
\end{subfigure}~~%
\begin{subfigure}{.48\textwidth}
  \vskip.07in
  \begin{center}
  \begin{tikzpicture}[x=0.6cm,y=0.6cm]
  \draw[lightgray] (5.48483,-0.5) -- (7.50555,-4);
\draw[lightgray] (0,-5) -- (2.59808,-0.5);
\draw[lightgray] (0,-5) -- (6.9282,-5);
\draw[lightgray] (1.73205,-2) -- (3.4641,-5);
\draw[lightgray] (1.73205,-2) -- (8.66025,-2);
\draw[lightgray] (2.02073,-1.5) -- (4.04145,-5);
\draw[lightgray] (2.02073,-1.5) -- (8.94893,-1.5);
\draw[lightgray] (6.06218,-0.5) -- (7.79423,-3.5);
\draw[lightgray] (0.57735,-4) -- (1.1547,-5);
\draw[lightgray] (0.57735,-4) -- (7.50555,-4);
\draw[lightgray] (3.17543,-0.5) -- (5.7735,-5);
\draw[lightgray] (4.33013,-0.5) -- (6.9282,-5);
\draw[lightgray] (5.19615,-5) -- (7.79423,-0.5);
\draw[lightgray] (6.35085,-5) -- (8.94893,-0.5);
\draw[lightgray] (4.90748,-0.5) -- (7.21688,-4.5);
\draw[lightgray] (6.9282,-5) -- (9.52628,-0.5);
\draw[lightgray] (0.866025,-3.5) -- (1.73205,-5);
\draw[lightgray] (0.866025,-3.5) -- (7.79423,-3.5);
\draw[lightgray] (4.04145,-5) -- (6.63953,-0.5);
\draw[lightgray] (1.44338,-2.5) -- (2.88675,-5);
\draw[lightgray] (1.44338,-2.5) -- (8.37158,-2.5);
\draw[lightgray] (2.59808,-0.5) -- (5.19615,-5);
\draw[lightgray] (2.59808,-0.5) -- (9.52628,-0.5);
\draw[lightgray] (3.75278,-0.5) -- (6.35085,-5);
\draw[lightgray] (4.6188,-5) -- (7.21688,-0.5);
\draw[lightgray] (5.7735,-5) -- (8.37158,-0.5);
\draw[lightgray] (2.88675,-5) -- (5.48483,-0.5);
\draw[lightgray] (8.94893,-0.5) -- (9.2376,-1);
\draw[lightgray] (3.4641,-5) -- (6.06218,-0.5);
\draw[lightgray] (2.3094,-1) -- (4.6188,-5);
\draw[lightgray] (2.3094,-1) -- (9.2376,-1);
\draw[lightgray] (0.288675,-4.5) -- (0.57735,-5);
\draw[lightgray] (0.288675,-4.5) -- (7.21688,-4.5);
\draw[lightgray] (1.1547,-3) -- (2.3094,-5);
\draw[lightgray] (1.1547,-3) -- (8.0829,-3);
\draw[lightgray] (0.57735,-5) -- (3.17543,-0.5);
\draw[lightgray] (1.73205,-5) -- (4.33013,-0.5);
\draw[lightgray] (6.63953,-0.5) -- (8.0829,-3);
\draw[lightgray] (1.1547,-5) -- (3.75278,-0.5);
\draw[lightgray] (7.79423,-0.5) -- (8.66025,-2);
\draw[lightgray] (2.3094,-5) -- (4.90748,-0.5);
\draw[lightgray] (7.21688,-0.5) -- (8.37158,-2.5);
\draw[lightgray] (8.37158,-0.5) -- (8.94893,-1.5);
\draw[line width=0.2mm] (4.80972,-0.809081) -- (4.42788,-1.19092);
\draw[line width=0.2mm] (4.42788,-0.809081) -- (4.80972,-1.19092);
\draw[line width=0.2mm] (4.80972,-1.80908) -- (4.42788,-2.19092);
\draw[line width=0.2mm] (4.42788,-1.80908) -- (4.80972,-2.19092);
\draw[line width=0.2mm] (5.0984,-0.309081) -- (4.71656,-0.690919);
\draw[line width=0.2mm] (4.71656,-0.309081) -- (5.0984,-0.690919);
\draw[line width=0.2mm] (5.0984,-1.30908) -- (4.71656,-1.69092);
\draw[line width=0.2mm] (4.71656,-1.30908) -- (5.0984,-1.69092);
\draw[black, line width=0.2mm, fill=white] (5.19615,-1) circle (0.216);
\draw[black, line width=0.2mm, fill=white] (5.19615,-2) circle (0.216);
\draw[black, line width=0.2mm, fill=white] (5.48483,-0.5) circle (0.216);
\draw[black, line width=0.2mm, fill=white] (5.48483,-1.5) circle (0.216);
\draw[line width=0.2mm] (5.67575,-2.30908) -- (5.29391,-2.69092);
\draw[line width=0.2mm] (5.29391,-2.30908) -- (5.67575,-2.69092);
\draw[black, line width=0.2mm, fill=white] (5.7735,-1) circle (0.216);
\draw[black, line width=0.2mm, fill=white] (5.7735,-2) circle (0.216);
\draw[black, line width=0.2mm, fill=white] (5.7735,-5) circle (0.252);
\draw[black, line width=0.16mm] (5.7735,-5) circle (0.18);
\draw[black, line width=0.2mm, fill=white] (6.06218,-0.5) circle (0.216);
\draw[black, line width=0.2mm, fill=white] (6.06218,-1.5) circle (0.216);
\draw[black, line width=0.2mm, fill=white] (6.06218,-2.5) circle (0.216);
\draw[black, line width=0.2mm, fill=white] (6.06218,-4.5) circle (0.252);
\draw[black, line width=0.16mm] (6.06218,-4.5) circle (0.18);
\draw[black, line width=0.2mm, fill=white] (6.35085,-1) circle (0.216);
\draw[black, line width=0.2mm, fill=white] (6.35085,-2) circle (0.216);
\draw[line width=0.2mm] (6.54177,-2.80908) -- (6.15993,-3.19092);
\draw[line width=0.2mm] (6.15993,-2.80908) -- (6.54177,-3.19092);
\draw[black, line width=0.2mm, fill=white] (6.35085,-4) circle (0.252);
\draw[black, line width=0.16mm] (6.35085,-4) circle (0.18);
\draw[black, line width=0.2mm, fill=white] (6.35085,-5) circle (0.252);
\draw[black, line width=0.16mm] (6.35085,-5) circle (0.18);
\draw[black, line width=0.2mm, fill=white] (6.63953,-0.5) circle (0.216);
\draw[black, line width=0.2mm, fill=white] (6.63953,-1.5) circle (0.216);
\draw[black, line width=0.2mm, fill=white] (6.63953,-2.5) circle (0.216);
\draw[black, line width=0.2mm, fill=white] (6.63953,-3.5) circle (0.252);
\draw[black, line width=0.16mm] (6.63953,-3.5) circle (0.18);
\draw[black, line width=0.2mm, fill=white] (6.63953,-4.5) circle (0.252);
\draw[black, line width=0.16mm] (6.63953,-4.5) circle (0.18);
\draw[black, line width=0.2mm, fill=white] (6.9282,-1) circle (0.216);
\draw[black, line width=0.2mm, fill=white] (6.9282,-2) circle (0.216);
\draw[black, line width=0.2mm, fill=white] (6.9282,-3) circle (0.252);
\draw[black, line width=0.16mm] (6.9282,-3) circle (0.18);
\draw[black, line width=0.2mm, fill=white] (6.9282,-4) circle (0.252);
\draw[black, line width=0.16mm] (6.9282,-4) circle (0.18);
\draw[fill=black] (6.9282,-5) circle (0.216);
\draw[black, line width=0.2mm, fill=white] (7.21688,-0.5) circle (0.216);
\draw[black, line width=0.2mm, fill=white] (7.21688,-1.5) circle (0.216);
\draw[black, line width=0.2mm, fill=white] (7.21688,-2.5) circle (0.252);
\draw[black, line width=0.16mm] (7.21688,-2.5) circle (0.18);
\draw[black, line width=0.2mm, fill=white] (7.21688,-3.5) circle (0.252);
\draw[black, line width=0.16mm] (7.21688,-3.5) circle (0.18);
\draw[fill=black] (7.21688,-4.5) circle (0.216);
\draw[black, line width=0.2mm, fill=white] (7.50555,-1) circle (0.216);
\draw[black, line width=0.2mm, fill=white] (7.50555,-2) circle (0.252);
\draw[black, line width=0.16mm] (7.50555,-2) circle (0.18);
\draw[black, line width=0.2mm, fill=white] (7.50555,-3) circle (0.252);
\draw[black, line width=0.16mm] (7.50555,-3) circle (0.18);
\draw[fill=black] (7.50555,-4) circle (0.216);
\draw[black, line width=0.2mm, fill=white] (7.79423,-0.5) circle (0.216);
\draw[black, line width=0.2mm, fill=white] (7.79423,-1.5) circle (0.252);
\draw[black, line width=0.16mm] (7.79423,-1.5) circle (0.18);
\draw[black, line width=0.2mm, fill=white] (7.79423,-2.5) circle (0.252);
\draw[black, line width=0.16mm] (7.79423,-2.5) circle (0.18);
\draw[fill=black] (7.79423,-3.5) circle (0.216);
\draw[black, line width=0.2mm, fill=white] (8.0829,-1) circle (0.252);
\draw[black, line width=0.16mm] (8.0829,-1) circle (0.18);
\draw[black, line width=0.2mm, fill=white] (8.0829,-2) circle (0.252);
\draw[black, line width=0.16mm] (8.0829,-2) circle (0.18);
\draw[fill=black] (8.0829,-3) circle (0.216);
\draw[black, line width=0.2mm, fill=white] (8.37158,-0.5) circle (0.252);
\draw[black, line width=0.16mm] (8.37158,-0.5) circle (0.18);
\draw[black, line width=0.2mm, fill=white] (8.37158,-1.5) circle (0.252);
\draw[black, line width=0.16mm] (8.37158,-1.5) circle (0.18);
\draw[fill=black] (8.37158,-2.5) circle (0.216);
\draw[black, line width=0.2mm, fill=white] (8.66025,-1) circle (0.252);
\draw[black, line width=0.16mm] (8.66025,-1) circle (0.18);
\draw[fill=black] (8.66025,-2) circle (0.216);
\draw[black, line width=0.2mm, fill=white] (8.94893,-0.5) circle (0.252);
\draw[black, line width=0.16mm] (8.94893,-0.5) circle (0.18);
\draw[fill=black] (8.94893,-1.5) circle (0.216);
\draw[fill=black] (9.2376,-1) circle (0.216);
\draw[fill=black] (9.52628,-0.5) circle (0.216);
  \end{tikzpicture}
  \end{center}
  \caption{Representation of layer sequence $\layerseq^{\stran}$ (excluding the crosses). Particles from $(\ntop_k)_{k \in [h]}$ are shifted upward, and two full columns are created from the $26$ particles in $(\nbot_k)_{k \in [h]}$, with $6$ particles from $(\nbot_k)_{k \in [h]}$ unused. These $6$ particles are added back to the positions marked with crosses to form the layer sequence $\layerseq^{\spost}$.}
  \label{fig:transform4}
\end{subfigure}%
\caption{Transformation from the layer sequence $\layerseq$ to $\layerseq^{\spost}$, with the layer sequences represented as right-justified configurations.}
\label{fig:transformation}
\end{figure*}

To understand the following transformations from $\layerseq$ to $\layerseq^{\spre}$ to $\layerseq^{\stran}$ to $\layerseq^{\spost}$, we can visualize layer sequences as ``right-justified'' configurations in $\Lambda$ (see Figure~\ref{fig:transformation}). The particles of these right-justified configurations can be grouped into columns, where the site on the $k$\textsuperscript{th} layer of column $j$ (counted from the right) is filled if and only if its layer sequence has at least $j$ particles in layer $k$.

The pre-processing step, which creates an intermediate layer sequence $\layerseq^{\spre} = (n^{\spre}_k)_{k \in [h]} \in \LayerSet^{0,h}$, is solely to account for specific edge cases in the eventual transformation to $\layerseq^{\spost}$.
To define $\layerseq^{\spre}$ starting from $\layerseq$, we first add in the unused $n-|\sigma|$ particles to the layer sequence. We denote $u = w - \max_{k \in [h]}\{n_k\} - 1$, the number of empty columns that we can fill with unused particles. We thus take $\min\left\{u, \floor{\frac{n-|\sigma|}{h}}\right\}$ empty columns and fill them with unused particles (which is equivalent to increasing each entry of the layer sequence by said amount of columns). The next step is to guarantee at least one completely filled column. If there are no completely filled columns even after adding the unused particles, $h-\toplayer$ particles are taken one at the time from the highest layer with at least two particles, and placed to form a single column of particles. This transformation is illustrated in Figure~\ref{fig:transform3}.

The particles of the layer sequence $\layerseq^{\spre}$ can be divided into three layer sequences $(\nfull_k)_{k \in [h]}$, $(\nbot_k)_{k \in [h]}$ and $(\ntop_k)_{k \in [h]}$ by column (so that $n^{\spre}_k = \nfull_k + \nbot_k + \ntop_k$ for all $k \in [h]$). 
The first layer sequence $(\nfull_k)_{k \in [h]}$ is represents the particles belonging to the $\Jfull = n^{\spre}_h \geq 1$ full columns of the right-justified configuration. 
The remaining columns are then identified as either \emph{bottom-supported} or \emph{non-bottom-supported}.
The bottom-supported columns are those which have some value $k' \in \{0,1,2,\ldots,h-1\}$ where layers $\{1,2,\ldots,k'\}$ in the column are occupied and the remaining layers are unoccupied (an alternative definition is having no unoccupied site below an occupied site). 
Particles from bottom-supported columns are put into $(\nbot_k)_{k \in [h]}$, while those from non-bottom-support columns go into $(\ntop_k)_{k \in [h]}$.
Alternatively, for each value of $k$ in $[h]$ we can define formally
\begin{itemize}
\item $\nfull_k = \Jfull$,
\item $\nbot_k = \Big|\Big\{j \in [w]: j > \Jfull \text{ and } n^{\spre}_{k'+1} \geq j \implies n^{\spre}_{k'} \geq j \text{ for all } k' \in [h-1]\Big\}\Big|$, and
\item $\ntop_k = n^{\spre}_k - \nbot_k - \nfull_k$.
\end{itemize}

To construct $\layerseq^{\stran}$, we apply different transformations to the particles corresponding to each of these layer sequences to construct a new configuration $\sigma^{\stran}$.
The particles in $(\nbot_k)_{k \in [h]}$ are transformed into $\Jbot = \floor{\frac{1}{h}\sum_{k=1}^h \nbot_k}$ full columns while the particles in $(\ntop_k)_{k \in [h]}$ have their particles shifted to the tops of their respective columns. These non-bottom-supported columns are placed directly next to the $\Jfull + \Jbot$ filled columns in the right-justified configuration $\sigma^{\stran}$ (see Figure~\ref{fig:transform4}). Note that there may be up to $h-1$ unused particles from $(\nbot_k)_{k \in [h]}$ after this transformation. The layer sequence of this new configuration will be denoted by $\layerseq^{\stran} = (n^{\stran}_k)_{k \in [h]}$.

To describe this transformation more precisely in terms of layer sequences, for each $i \in [w]$, we denote $\coltop_i = |\{k \in [h]: \ntop_k \geq i\}|$, the number of particles in the $i$\textsuperscript{th} non-bottom-supported column. With $\Jfull$ and $\Jbot$ defined as before, we can then define the layer sequence $\layerseq^{\stran} = (n^{\stran}_k)_{k \in [h]}$, where for each $k \in [h]$,
\[n^{\stran}_k = \Jfull + \Jbot + \Big|\Big\{i \in [w] : \coltop_i \geq h - k + 1\Big\}\Big|.\]
We note that this is a valid transformation (into a layer sequence $\layerseq^{\stran} \in \LayerSet^{0,h}$) as for each $k \in [h]$, we know that $\Jbot \leq \max_{k'}\{\nbot_{k'}\}$ and $\Big|\Big\{i \in [w] : \coltop_i \geq h - k + 1\Big\}\Big| \leq \max_{k'}\{\ntop_{k'}\}$, so we must have $n^{\stran}_k \leq n^{\spre}_k \leq w - 1$.

We include one final step in the transformation to ensure that the change in scent intensity from $\layerseq^{\spre}$ to $\layerseq^{\spost}$ is non-negative.
In particular, in the transformation from $\layerseq^{\spre}$ to $\layerseq^{\spost}$, for the sake of simplicity, we had transformed the particles from $(\nbot_k)_{k \in [h]}$ into $\Jbot = \floor{\frac{1}{h}\sum_{k=1}^h \nbot_k}$ full columns, leaving $m \leq h-1$ particles unused. As a final post-processing step, we can add these $m$ particles back to the final layer sequence, by adding one particle each to the topmost $m$ rows that have less than $w-1$ particles (note that the final layer sequence is non-decreasing). This gives a positive net change in scent intensity for the particles from $(\nbot_k)_{k \in [h]}$, while potentially increasing the value of $\Bhat_h$ by at most $O_h(1)$ (using the definition of $\Bhat_h$ in Lemma~\ref{lem:boundarylengthlowerbound}).

\subparagraph*{Change in scent intensity:} We compute the overall change in the two transformations to show Lemma~\ref{lem:bridgetransformscentintensity}. It is important to note that in our analysis, we do not assume that $\varphi$ is constant with respect to $h$.
\begin{lemma}[Scent Intensity]
\label{lem:bridgetransformscentintensity}
The changes in scent intensities in the transformations from $\layerseq$ to $\layerseq^{\spre}$ and from $\layerseq^{\spre}$ to $\layerseq^{\spost}$ are non-negative and have the following lower bounds:
\begin{align*}
\Shat(\layerseq^{\spre}) - \Shat(\layerseq) &\geq \varphi \cdot \min\left\{u, \floor[\Big]{\frac{n-\nin}{h}}\right\} \\
\Shat(\layerseq^{\spost}) - \Shat(\layerseq^{\spre}) &\geq \nin \frac{\varphi}{h}\left(1 - \frac{\toplayer}{h}\right) + \varphi \cdot O_h(1),
\end{align*}
where $u = w - \max_{k \in [h]}\{n_k\} - 1$ and $O_h(1)$ represents a function that does not depend on the specific choice of starting configuration $\sigma$.
\end{lemma}
%
\begin{proof}
For the transformation from $\layerseq$ to $\layerseq^{\spre}$, the number of full columns added to the $\layerseq$ is $\min\left\{u, \floor[\Big]{\frac{n-\nin}{h}}\right\}$, and neither the adding of full columns nor the shifting of $h-D$ particles to the top $h-D$ rows can decrease the total scent intensity.

The transformation from $\layerseq^{\spre}$ to $\layerseq^{\spost}$ requires a more involved analysis. The construction of $\layerseq^{\spost}$ from $\layerseq^{\stran}$ explains why $\Shat(\layerseq^{\spost}) - \Shat(\layerseq^{\spre}) \geq 0$, but for the main lower bound, we focus on the transformation from $\layerseq^{\spre}$ to $\layerseq^{\stran}$, while noting that $\Shat(\layerseq^{\spost}) - \Shat(\layerseq^{\stran}) \geq 0$.

In the transformation from $\layerseq^{\spre}$ to $\layerseq^{\stran}$,
we will make the claim that for particles from both bottom-supported columns and non-bottom-supported columns, the average increase in the total scent intensity per particle has a lower bound of $\frac{\varphi}{h}\left(1 - \frac{\toplayer}{h}\right)$, with a small error term.
To show this, let $\Scont$ be a convex function defined from $\Sh$ using Lemma~\ref{lem:convexlemma}
and note that $\varphi = \sum_{k=1}^h \Sh(k) = \int_0^h \Scont(y)dy$.
%
%
For the particles from $(\nbot_k)_{k \in [h]}$, a lower bound for the total increase in scent intensity from $\layerseq^{\spre}$ to $\layerseq^{\stran}$ is given by
%
%
\begingroup
\allowdisplaybreaks
\begin{align*}
\Jbot\sum_{k=1}^h \Sh(k) - \sum_{i \in [w]}\sum_{k=1}^{\colbot_i}\Sh(k)
&= \left(\sum_{i \in [w]} \colbot_i + O(h)\right) \cdot \frac{1}{h}\sum_{k=1}^h \Sh(k) - \sum_{i \in [w]}\colbot_i \frac{1}{\colbot_i} \sum_{k=1}^{\colbot_i}\Sh(k) \\
&= \sum_{i \in [w]} \colbot_i\left(\frac{1}{h}\sum_{k=1}^h \Sh(k) - \frac{1}{\colbot_i} \sum_{k=1}^{\colbot_i}\Sh(k) \right) + \varphi \cdot O_h(1) \\
&\geq \sum_{i \in [w]} \colbot_i \left( \frac{1}{h}\int_{0}^{h}\Scont(y)dy - \frac{1}{\colbot_i}\int_{0}^{\colbot_i}\Scont(y)dy \right) + \varphi \cdot O_h(1) \\
&= \sum_{i \in [w]} \colbot_i \int_{0}^{1}\left(\Scont(yh) - \Scont(y\colbot_i) \right)dy + \varphi \cdot O_h(1) \\
&\geq \left(\sum_{k \in [h]} \nbot_k\right) \int_{0}^{1}\left(\Scont(yh) - \Scont(y\toplayer) \right)dy + \varphi \cdot O_h(1).
\end{align*}
\endgroup

For the particles from $(\ntop_k)_{k \in [h]}$, to construct a lower bound for the total increase in scent intensity from $\layerseq^{\spre}$ to $\layerseq^{\stran}$, we denote $\Delta_i = \Delta \Sh(\toplayer - \colbot_i) = \Sh(\toplayer - \colbot_i + 1) - \Sh(\toplayer - \colbot_i)$. We construct this lower bound by combining two lower bounds, the first of which is:
\begin{align*}
&\sum_{i \in [w]} \sum_{k=h-\coltop_i+1}^h \Sh(k) - \sum_{i \in [w]} \sum_{k=\toplayer-\coltop_i+1}^\toplayer \Sh(k) 
= \sum_{i \in [w]} \sum_{k=\toplayer-\coltop_i+1}^\toplayer \left(\Sh(k+h-\toplayer) - \Sh(k)\right) \\
&= \sum_{i \in [w]} \sum_{k=\toplayer-\coltop_i+1}^\toplayer \sum_{j=k}^{k+h-\toplayer-1} \Delta \Sh(j) 
\geq \sum_{i \in [w]} \coltop_i (h-\toplayer-1) \Delta_i
\geq (h-\toplayer-1) \sum_{i \in [w]} \Delta_i.
\end{align*}
The second lower bound we construct also makes use of Lemma~\ref{lem:convexlemma}:
%
%
%
\begin{align*}
&\sum_{i \in [w]} \sum_{k=h-\coltop_i+1}^h \Sh(k) - \sum_{i \in [w]} \sum_{k=\toplayer-\coltop_i+1}^\toplayer \Sh(k) \\
&\geq \sum_{i \in [w]} \left( \coltop_i \left(\frac{1}{\coltop_i}\int_{0}^{\coltop_i} \Scont(h-\coltop_i+y)dy - \frac{1}{\coltop_i}\int_{0}^{\coltop_i} \Scont(\toplayer-\coltop_i+y) dy \right) - \frac{1}{8}\Delta_i \right) \\
&\geq \sum_{i \in [w]} \left( \coltop_i \int_{0}^{1} \left(\Scont(\toplayer-\coltop_i+y\coltop_i + (h-\toplayer)) - \Scont(\toplayer-\coltop_i+y\coltop_i)\right) dy - \frac{1}{8}\Delta_i \right) \\
&\geq \sum_{i \in [w]} \left( \coltop_i \int_{0}^{1} \left(\Scont(y\toplayer + (h-\toplayer)) - \Scont(y\toplayer)\right) dy - \frac{1}{8}\Delta_i \right) \\
&\geq \left(\sum_{k \in [h]} \ntop_k\right) \int_{0}^{1} \left(\Scont(yh) - \Scont(y\toplayer)\right) dy - \sum_{i \in [w]}\Delta_i,
\end{align*}
where the second last line is due to $\Scont$ being convex and thus
having a nondecreasing subgradient, and so for any $x \in (0,h)$ and $a_1,a_2 \in [0,h-x]$, $a_1 < a_2$, we have $\Scont(a_1+x) - \Scont(a_1) \leq \Scont(a_2+x) - \Scont(a_2)$.
By using the former bound in the case where $(h-\toplayer-1)\sum_{i \in [w]}\Delta_i \geq \left(\sum_{k \in [h]} \ntop_k\right) \int_{0}^{1} \left(\Scont(yh) - \Scont(y\toplayer)\right) dy$ and the second lower bound in the case where $(h-\toplayer-1)\sum_{i \in [w]}\Delta_i < \left(\sum_{k \in [h]} \ntop_k\right) \int_{0}^{1} \left(\Scont(yh) - \Scont(y\toplayer)\right) dy$, we can conclude that:
\begin{align*}
&\sum_{i \in [w]} \sum_{k=h-\coltop_i+1}^h \Sh(k) - \sum_{i \in [w]} \sum_{k=\toplayer-\coltop_i+1}^\toplayer \Sh(k) \\
&\geq \left(\sum_{k \in [h]} \ntop_k\right) \left(\int_{0}^{1} \left(\Scont(yh) - \Scont(y\toplayer)\right) dy \right)\left(1 - \frac{1}{h-\toplayer-1}\right) \\
&\geq \left(\sum_{k \in [h]} \ntop_k\right) \int_{0}^{1} \left(\Scont(yh) - \Scont(y\toplayer)\right) dy + \varphi \cdot O_h(1).
\end{align*}
~

\noindent
As $\Scont$ has a non-decreasing gradient, we have
\[\Scont(y) \leq y\frac{\Scont(\toplayer)}{\toplayer} \text{ for } y \in [0,\toplayer] \text{ and }
\Scont(y) \geq y\frac{\Scont(\toplayer)}{\toplayer} \text{ for } y \in [\toplayer,h],\]
which gives us
\[\int_0^{\toplayer} \Scont(y)dy \leq \frac{\Scont(\toplayer)}{\toplayer}\cdot \frac{\toplayer^2}{2} \text{ and }
\int_{\toplayer}^{h} \Scont(y)dy \geq \frac{\Scont(\toplayer)}{\toplayer}\cdot \frac{h^2 - \toplayer^2}{2}.\]
thus
\begin{align*}
\frac{\int_{0}^{h} \Scont(y)dy}{\int_{0}^{\toplayer} \Scont(y)dy}
= 1 + \frac{\int_{\toplayer}^{h} \Scont(y)dy}{\int_{0}^{\toplayer} \Scont(y)dy}
\geq 1 + \frac{h^2 - \toplayer^2}{\toplayer^2}
= \frac{h^2}{\toplayer^2},
\end{align*}
which gives us a lower bound on the following integral
\begin{align*}
\int_{0}^{1} \left(\Scont(yh) - \Scont(y\toplayer)\right) dy 
&=\frac{1}{h}\int_{0}^{h} \Scont(y)dy - \frac{1}{\toplayer}\int_{0}^{\toplayer} \Scont(y)dy\\
&\geq \frac{1}{h}\int_{0}^{h} \Scont(y)dy - \frac{\toplayer^2}{\toplayer h^2} \int_{0}^{h} \Scont(y)dy\\
&= \frac{\varphi}{h}\left(1 - \frac{\toplayer}{h}\right),
\end{align*}
which will give us lower bounds for the per-particle increase in scent intensity from $\layerseq^{\spre}$ to $\layerseq^{\stran}$ for particles from both $(\nbot_k)_{k \in [h]}$ and $(\ntop_k)_{k \in [h]}$ (which comprises all but $h-1$ particles from $\sigma$).
This gives us the lower bound for $\Shat(\layerseq^{\stran}) - \Shat(\layerseq^{\spre})$ from the lemma.
\end{proof}

\subparagraph*{Change in boundary length:} Instead of computing the boundary lengths for configurations directly, we compare the (tight) boundary length lower bounds that are defined in Lemma~\ref{lem:boundarylengthlowerbound} between the initial and final layer sequences $\layerseq$ and $\layerseq^{\spost}$.
\begin{lemma}[Boundary Length]
\label{lem:bridgetransformboundarylength}
In the transformation from $\layerseq \in \LayerSet^{0,D}$ to $\layerseq^{\spost} \in \LayerSet^{0,h}$, we have:
\begin{align*}
\Bhat_h(\layerseq^{\spost}) - \Bhat_D(\layerseq)
\leq 4(h-\toplayer) - 2\max\left\{w-u-D, \left(\frac{\nin}{h} - \frac{h}{2}\right)\right\} + O_h(1),
\end{align*}
where $O_h(1)$ represents a function that does not depend on the specific choice of starting configuration $\sigma$.
Furthermore, by Lemma~\ref{lem:boundarylengthlowerbound} there exists a configuration $\tau$
with layer sequence $\layerseq^{\spost}$ such that $B(\tau) - B(\sigma) \leq \Bhat_h(\layerseq^{\spost}) - \Bhat_D(\layerseq)$, and thus has the same upper bound.
\end{lemma}
%
\begin{proof}
Our analysis focuses almost entirely on the comparison between layer sequences $\layerseq$ and $\layerseq^{\stran}$, by noting that the post-processing step produces an insignificant change in boundary length, in the sense that:
\[\Bhat_h(\layerseq^{\spost}) = \Bhat_h(\layerseq^{\stran}) + O_h(1).\]

\noindent
To analyze the difference between $\Bhat_D(\layerseq)$ and $\Bhat_h(\layerseq^{\stran})$, we begin with the following definitions:
\begin{multicols}{2}
\begin{itemize}
\item $\extminus = \sum_{k=1}^{\toplayer-1} \max\{n_{k+1} - n_k - 1, 0\}$
\item $\extplus = \sum_{k=1}^{\toplayer-1} \max\{n_{k+1} - n_k + 1, 0\}$
\item $\extminus^{\stran} = \sum_{k=1}^{h-1} \max\{n^{\stran}_{k+1} - n^{\stran}_k - 1, 0\}$
\item $\extplus^{\stran} = \sum_{k=1}^{h-1} \max\{n^{\stran}_{k+1} - n^{\stran}_k + 1, 0\}$
\end{itemize}
\end{multicols}
\noindent
This allows us to write $\Bhat_D(\layerseq) = 2n_1 + 2 + 2\toplayer + 2\extminus + 2\extplus$ and $\Bhat_h(\layerseq^{\stran}) = 2n^{\stran}_1 + 2 - 2n^{\stran}_h + 2h + 2\extminus^{\stran} + 2\extplus^{\stran}$ (referring to their expressions in Lemma~\ref{lem:boundarylengthlowerbound}).
%
%
To lower bound $\Bhat_D(\layerseq)$, we first note that
\begin{align*}
\nin
&= \sum_{m=1}^\toplayer(n_m + n_1 - n_1)
= \toplayer n_1 + \sum_{m=1}^\toplayer \sum_{k=1}^{m-1}(n_{k+1} - n_k + 1 - 1) \\
&= \toplayer n_1 - \left(\frac{\toplayer^2}{2} - \frac{\toplayer}{2}\right) + \sum_{m=1}^\toplayer \sum_{k=1}^{m-1}(n_{k+1} - n_k + 1) \\
&\leq \toplayer n_1 - \left(\frac{\toplayer^2}{2} - \frac{\toplayer}{2}\right) + \sum_{m=1}^\toplayer \extplus = (n_1 + \extplus)\toplayer - \left(\frac{\toplayer^2}{2} - \frac{\toplayer}{2}\right).
\end{align*}
In addition, we also have that
\begin{align*}
n_1 + \extplus = n_1 + \sum_{k=1}^{\toplayer-1} \max\{n_{k+1} - n_k + 1, 0\} \geq \max_{k \in [h]} \{n_k\} \geq w - u - 1.
\end{align*}
This means that
\begin{align*}
\Bhat_D(\layerseq)
&\geq 2n_1 + 2\toplayer + 2\extminus + 2\extplus \\
&\geq 2\toplayer + 2\extminus + 2\max\left\{w-u-1, \left(\frac{\nin}{\toplayer} + \frac{\toplayer}{2} - \frac{1}{2}\right)\right\} \\
&= 4\toplayer + 2\extminus + 2\max\left\{w-u-\toplayer, \left(\frac{\nin}{\toplayer} - \frac{\toplayer}{2}\right)\right\} + O_h(1)\\
&\geq 4\toplayer + 2\extminus + 2\max\left\{w-u-\toplayer, \left(\frac{\nin}{h} - \frac{h}{2}\right)\right\} + O_h(1).
\end{align*}

\noindent
We then compute an upper bound on $\Bhat_h(\layerseq^{\stran})$:
\begin{align*}
\Bhat_h(\layerseq^{\stran})
&= 2n^{\stran}_1 - 2n^{\stran}_h + 2h + 2\extminus^{\stran} + 2\extplus^{\stran} \\
&= 2h + 2\extminus^{\stran} + 2\sum_{k=1}^{h-1}\left(\max\{n^{\stran}_{k+1} - n^{\stran}_k + 1, 0\} - (n^{\stran}_{k+1} - n^{\stran}_k)\right) \\
&\leq 2h + 2\extminus^{\stran} + 2h \text{ (as $n^{\stran}_{k+1} - n^{\stran}_k \geq 0$ for all $k$).}
\end{align*}

Recalling the intermediate layer sequence $\layerseq^{\spre}$ and defining $\extminus^{\spre} = \sum_{k=1}^{h-1} \max\{n^{\spre}_{k+1} - n^{\spre}_k - 1, 0\}$, we show that $\extminus^{\stran} \leq \extminus^{\spre} \leq \extminus$. It is easy to check that $\extminus^{\spre} \leq \extminus$, as we have $\extminus = \sum_{k=1}^{\toplayer-1} \max\{n_{k+1} - n_k - 1, 0\} = \sum_{k=1}^{h-1} \max\{n_{k+1} - n_k - 1, 0\}$ because $\layerseq \in \LayerSet^{0,D}$, and all actions done in the construction of $\layerseq^{\spre}$ from $\layerseq$ do not increase the value of $\extminus$.

To show that $\extminus^{\stran} \leq \extminus^{\spre}$, we make the observation that as $n^{\stran}_{k+1} \geq n^{\stran}_k$ for all layers $k$, the contributions to $\extminus^{\stran}$ come exactly from pairs of adjacent non-empty columns of $\sigma^{\stran}$ with the same number of particles:
\[\extminus^{\stran} = \sum_{k=1}^{h-1} \max\{n^{\stran}_{k+1} - n^{\stran}_k - 1, 0\} = |\{i \in [w]: \coltop_i \geq 1 \text{ and } \coltop_i = \coltop_{i-1}\}|.\]
For any column $i \geq 2$ where $\coltop_i \geq 1$ and $\coltop_i = \coltop_{i-1}$, we must have $\ntop_{k} \geq i-1 \iff \ntop_{k} \geq i$ for all $k \in [h]$. 

Columns $i$ and $i-1$ in $\layerseq^{\stran}$ correspond to a pair of non-bottom-supported columns in $\layerseq^{\spre}$. As the property $\ntop_{k} \geq i-1 \iff \ntop_{k} \geq i$ necessitates that these columns be adjacent, we will label the columns in $\layerseq^{\spre}$ that they correspond to $j_i$ and $j_i-1$ respectively. Due to being non-bottom-supported, there must be a row $k$ where $n_k \geq j_i$ and $n_{k-1} \leq j_i - 2$, meaning that this row and pair of columns contributes to the sum $\extminus^{\spre} = \sum_{k=1}^{\toplayer-1} \max\{n^{\spre}_{k+1} - n^{\spre}_k - 1, 0\}$. The corresponding pair of columns is unique to each choice of $i, i-1$, which allows us to conclude that
\[|\{i \in [w]: \coltop_i \geq 1 \text{ and } \coltop_i = \coltop_{i-1}\}| \leq \extminus^{\spre}.\]

Thus, we can conclude the proof as follows:
\begin{align*}
&\Bhat_h(\layerseq^{\stran}) - \Bhat_D(\layerseq)\\
&\leq \left(4h + 2\extminus^{\stran}\right) - \left(4\toplayer + 2\extminus + 2\max\left\{w-u-D, \left(\frac{\nin}{h} - \frac{h}{2}\right)\right\}\right) + O_h(1)\\
&\leq 4(h-\toplayer) - 2\max\left\{w-u-D, \left(\frac{\nin}{h} - \frac{h}{2}\right)\right\} + O_h(1). \qedhere
\end{align*}
\end{proof}

\subparagraph*{Change in the Hamiltonian:}
Finally, putting together the changes in scent intensity (Lemma~\ref{lem:bridgetransformscentintensity}) and boundary length (Lemma~\ref{lem:bridgetransformboundarylength}) going from $\layerseq$ to $\layerseq^{\spost}$, we can compute a lower bound for the change in the Hamiltonian.
%
%
\begin{proof}[Proof of Lemma~\ref{lem:thm2case1}]
We consider two cases, when $\floor[\big]{\frac{n-\nin}{h}} > u$ and when $\floor[\big]{\frac{n-\nin}{h}} \leq u$, roughly corresponding to whether or not there were enough empty columns in $\layerseq$ to fit the majority of the unused particles in the transformation from $\layerseq$ to $\layerseq^{\spre}$.
The easier case is when $\floor[\big]{\frac{n-\nin}{h}} > u$, in which case we have
\begin{align*}
H(\sigma) - H(\tau) 
&\geq \Bhat_D(\layerseq) - \Bhat_h(\layerseq^{\spost}) + \eta\left(\Shat(\layerseq^{\spost}) - \Shat(\layerseq^{\spre})\right) + \eta\left(\Shat(\layerseq^{\spre}) - \Shat(\layerseq)\right) \\
&\geq \Big(2(w-u-\toplayer) - 4(h-\toplayer) + O_h(1)\Big) + \eta\cdot 0 + \eta \varphi u \\
&\geq 2D + 2w - 4h + u(\eta \varphi - 2) + O_h(1) \\
&\geq 2D + 2\rho h + u(\eta \varphi - 2) + O_h(1).
\end{align*}
Having $\eta > \frac{4}{\varphi}\left(1 + \frac{1}{\rho}\right)$ implies that $\eta \varphi - 2 \geq 0$, so we can conclude that
\[
H(\sigma) - H(\tau) \geq 2\rho h + O_h(1).
\]

We now establish a similar bound for the case when $\floor[\big]{\frac{n-\nin}{h}} \leq u$. We have
\begin{align*}
&H(\sigma) - H(\tau) \\
&\geq \Bhat_D(\layerseq) - \Bhat_h(\layerseq^{\spost}) + \eta\left(\Shat(\layerseq^{\spost}) - \Shat(\layerseq^{\spre})\right) + \eta\left(\Shat(\layerseq^{\spre}) - \Shat(\layerseq)\right) \\
&\geq 2\left(\frac{\nin}{h} - \frac{h}{2}\right) - 4(h-\toplayer) + \eta \varphi \frac{\nin}{h}\left(1 - \frac{\toplayer}{h}\right) + \eta \varphi \frac{n-\nin}{h} + \varphi \cdot O_h(1) \\
&\geq 2\left(\frac{n}{h} - \frac{h}{2}\right) - 2\frac{n-\nin}{h}-4(h-\toplayer) + \eta \varphi \frac{n}{h}\left(1 - \frac{\toplayer}{h}\right) + \eta \varphi \frac{n-\nin}{h}\cdot \frac{\toplayer}{h} + \varphi \cdot O_h(1) \\
&= 2\left(\frac{n}{h} - \frac{h}{2}\right) - 4(h-\toplayer) + \eta \varphi \rho h \left(1 - \frac{\toplayer}{h}\right) + \frac{n-\nin}{h}\left(\eta \varphi \frac{\toplayer}{h} - 2\right) + \varphi \cdot O_h(1) \\
&= 2 h\left(\rho - \frac{1}{2}\right) + h\left(1-\frac{\toplayer}{h}\right)\left(\eta \varphi \rho - 4 \right)  + \frac{n-\nin}{h}\left(\eta \varphi \frac{\toplayer}{h} - 2\right) + \varphi \cdot O_h(1). \numberthis \label{eqn:formbridge}
\end{align*}
As $\rho > 0.5$, we have $2\left(\rho - \frac{1}{2}\right) > 0$. For the values of $\toplayer$ where $\eta \varphi \frac{\toplayer}{h} - 2 > 0$, as we must also have $\eta \varphi \rho - 4 > 0$, giving us our desired result. For the values of $\toplayer$ where $\eta \varphi \frac{\toplayer}{h} - 2 \leq 0$, we must have $1 - \frac{\toplayer}{h} \geq 1 - \frac{2}{\eta \varphi}$, which allows us to bound the second and third terms of~(\ref{eqn:formbridge}) as:
\begin{align*}
h\left(1-\frac{\toplayer}{h}\right)\left(\eta \varphi \rho - 4 \right)  + \frac{n-\nin}{h}\left(\eta \varphi \frac{\toplayer}{h} - 2\right)
&\geq h\left(1-\frac{2}{\eta\varphi}\right)\left(\eta\varphi\rho - 4\right) - \rho h \cdot 2 \\
&= h\left(\eta\varphi\rho - 4(1+\rho) + \frac{8}{\eta\rho} \right).
\end{align*}
Thus in this case we also have $H(\sigma) - H(\tau) \geq (2\rho - 1)h + O_h(1).$

It remains to rewrite this lower bound as $\epsilon h + \delta B(\sigma) + O_h(1)$ for suitable $\epsilon$ and $\delta$, as in the lemma statement. Denoting $\epsilon' = 2\rho - 1$ so that $H(\sigma) - H(\tau) \geq \epsilon' h + O_h(1)$, and denoting:
\[\delta = \frac{\epsilon'}{\epsilon' + 4+ 4\alpha},\text{ and } \epsilon = \epsilon' - \delta(\epsilon' + 4+ 2\alpha),\]
as $B(\tau) \leq 4h + 2\extplus^{\spost} \leq 4h + 2w$, we can see that 
\begin{align*}
H(\sigma) - H(\tau)
&\geq (1-\delta)\epsilon' h + \delta\left(\left(B(\sigma) - B(\tau)\right) - \eta\left(S(\sigma) - S(\tau)\right)\right) + O_h(1) \\
&\geq (1-\delta)\epsilon' h - \delta(4h + 2w) + \delta B(\sigma) 
= \epsilon h + \delta B(\sigma) + O_h(1).
\end{align*}

This essentially concludes the proof. However, for completeness, we address configurations with layer sequences in $\LayerSet^{R,\toplayer}$ for some $R > 0$. We do this by reducing to the case when the layer sequence is in $\LayerSet^{0,\toplayer}$, by removing a single column and bounding the associated changes in scent intensity and boundary length.

In this case, we start with a single pre-processing step to remove one particle from each of the completely filled rows of $\layerseq^0$, to transform it into the layer sequence $\layerseq = (n_k)_{k \in [h]} \in \LayerSet^{0,\toplayer}$ used in the main proof. We then add these $R$ particles back in a final post-processing step similar to that of the initial transformation, by adding one particle each to the topmost $R$ rows that have less than $w-1$ particles. This ensures that the net change in scent intensity in the transformation to the final configuration $\tau$ is still non-negative, while once again potentially increasing $\Bhat_h(\layerseq^{\spost})$ by at most $O_h(1)$.
To quantify the change in boundary length more accurately, we observe that:
\begin{align*}
\Bhat_D(\layerseq^0) = \Bhat_D(\layerseq) - 4R + 2.
\end{align*}
Recalling that $\Bhat_D(\layerseq) = 2n_1 + 2D + 2\extminus + 2\extplus$ from before, and that $\extplus \geq R-1$ as $n_k = w-1$ for all $k \in \{1,2,\ldots,R\}$, we can lower bound $\Bhat_D(\layerseq^0)$ as follows:
\begin{align*}
\Bhat_D(\layerseq^0) 
&= 2n_1 + 2D + 2\extminus + 2\extplus - 4R + 2\\
&\geq 2w - 2 + 2(D-R) + 2\extminus + 2(R-1) - 2R + 2 \geq 2w + 2\extminus - 2.
\end{align*}
Thus, we end up with a similar bound for the differences in the Hamiltonian:
\begin{align*}
H(\sigma^0) - H(\tau)
&= \left(B(\sigma^0) - B(\tau)\right) - \eta\left(S(\sigma^0) - S(\tau)\right) \\
&\geq \left(2w + 2\extminus-2\right) - \left(4h + 2\extminus^{\spost}\right) + O_h(1)
\geq 2\rho h + O_h(1).
\end{align*}
As $H(\sigma^0) - H(\tau) \geq \epsilon'h + O_h(1)$ like before, we can follow a similar argument to write the lower bound in the appropriate form.
\end{proof}

The proof of Lemma~\ref{lem:thm2case2} is shorter, though some of it relies on the transformation used in the proof of Lemma~\ref{lem:thm2case1}.

\begin{proof}[Proof of Lemma~\ref{lem:thm2case2}]
Let $\sigma$ be a configuration that reaches $\Lambda_h$ and let $\layerseq = (n_k)_{k \in [h]}$ be its layer sequence, and define $\extminus = \sum_{k=1}^{h-1} \max\{n_{k+1} - n_k - 1, 0\}$ and $\extplus = \sum_{k=1}^{h-1} \max\{n_{k+1} - n_k + 1, 0\}$.  
We can then apply the same transformation used from $\layerseq^{\spre}$ to $\layerseq^{\spost} = (n_k^{\spost})_{k \in [h]}$ in the proof of Lemma~\ref{lem:thm2case1} to obtain a new configuration $\tau$ with layer sequence $\layerseq^{\spost}$ where $S(\tau) \geq S(\sigma)$.
%
Denoting $\extminus^{\spost} = \sum_{k=1}^{h-1} \max\{n^{\spost}_{k+1} - n^{\spost}_k - 1, 0\}$,
we know from the proof of Lemma~\ref{lem:thm2case1} that $n^{\spost}_{k+1} \geq n^{\spost}_{k}$ for all layers $k \in [h-1]$, and $\extminus^{\spost}$ corresponds to the number of non-filled, non-empty pairs of adjacent columns of $\tau$ with the same number of particles. As there are at most $h$ possible column heights, the number of non-completely-filled, non-empty columns is at most $h + \extminus^{\spost}$. On the other hand, the number of completely filled columns can be upper bounded by $\frac{n}{h} = \rho h$. Each column gives contributes at most $\varphi$ to the total scent intensity of $\tau$, and combining this with the fact (from the proof of Lemma~\ref{lem:thm2case1}) that $\extminus^{\stran} \leq \extminus$ and additionally noting that $\extminus^{\spost} \leq \extminus^{\stran}+1$, this gives us
\[S(\sigma) \leq S(\tau) \leq \varphi \left(h(1+\rho) + \extminus^{\spost}\right) \leq \varphi \left(h(1+\rho) + \extminus + 1\right).\]
%

\noindent
By Lemma~\ref{lem:boundarylengthlowerbound}, the lower bound for the boundary length of $\sigma$ can be computed as
\begin{align*}
B(\sigma) 
&\geq \Bhat_h(\layerseq) 
= 2 + 2n_1 - 2n_h + 2\extminus + 2\extplus + 2h \\
&= 2\extminus + 2h + 2\sum_{k=1}^{h-1}\left(\max\{n_{k+1} - n_k + 1, 0\} - (n_{k+1}-n_k)\right)
= 2\extminus + 4h.
\end{align*}
%
We then compare the Hamiltonian $H(\sigma)$ to that of the empty configuration $\emptyset$. As $H(\emptyset) = 0$, for any $\delta \in (0,1)$ we have
\begin{align*}
H(\sigma) - H(\emptyset)
&\geq (1-\delta)\left(2\extminus + 4h\right) - \eta\varphi \left(h(1+\rho) + \extminus + 1\right) + \delta B(\sigma)\\
&= - h(\eta\varphi(1+\rho) - 4(1-\delta)) - \extminus(\eta\varphi - 2(1-\delta)) - \eta \varphi + \delta B(\sigma).
\numberthis \label{eqn:nobridgehamiltonian}
\end{align*}
%

It remains to write this lower bound as $\epsilon h + \delta B(\sigma) + O_h (1)$ for suitable $\epsilon$ and $\delta$, as in the lemma statement. We use the following conditions for $\beta$ and $\eta$:
\begin{align*}
\beta > \left(1 + \frac{\alpha}{2}\right)\log 2
\text{ and }
\eta < \frac{1}{\varphi}\min\left\{ 2\left(1 - \frac{\log 2}{\beta}\right), \frac{4}{1+\rho}\left( 1 - \left(1 + \frac{\alpha}{2}\right)\frac{\log 2}{\beta} \right) \right\}.
\end{align*}
Due to our choice of $\beta$, the above expression that $\eta$ must be smaller than must always be positive. In addition, our condition on $\eta$ ensures that $\eta \varphi = O_h(1)$. By setting
\begin{align*}
\delta = \min\left\{ \frac{4-\eta\varphi(1+\rho)}{4 + 2\alpha}, \frac{2-\eta\varphi}{2} \right\}
\text{ and }
\epsilon = 4(1-\delta)-\eta\varphi(1+\rho),
\end{align*}
we can see by~(\ref{eqn:nobridgehamiltonian}) that 
\[
H(\sigma) - H(\emptyset) \geq \epsilon h + \delta B(\sigma) + O_h(1).
\]
Note that $\delta$ must be a positive value as our choice of $\eta$ guarantees that $4 - \eta\varphi(1+\rho) > 0$ and $2 - \eta\varphi > 0$.
\end{proof}

\end{document}